\documentclass[12pt]{article}
\pdfoutput=1

\usepackage{amsmath,amssymb,  amsthm, bm, amsfonts, cite}
\usepackage[a4paper]{geometry}
\usepackage[colorlinks=true, pdfstartview=FitV, linkcolor=blue, citecolor=blue, urlcolor=blue]{hyperref}
\usepackage{graphicx}
\usepackage{subcaption}
\usepackage{xcolor}

\def\ran{\mathop{\rm Ran}\nolimits}
\newcommand {\bB}{{\mathbb B}}
\newcommand {\bT}{{\mathbb T}}
\newcommand {\bC}{{\mathbb C}}
\newcommand {\bZ}{{\mathbb Z}}
\newcommand {\bN}{{\mathbb N}}
\newcommand {\bI}{{\mathbb I}}
\newcommand {\bM}{{\mathbb M}}
\newcommand {\bR}{{\mathbb R}}

\newcommand {\hil}{{\mathcal H}}
\newcommand {\cH}{{\mathcal H}}
\newcommand {\cB}{{\mathcal B}}
\newcommand {\cC}{{\mathcal C}}

\newcommand {\cE}{{\mathcal E}}

\newcommand {\cP}{{\mathcal P}}

\newtheorem{thm}{Theorem} [section]
\newtheorem{theorem}[thm]{Theorem}
\newtheorem{lemma}[thm]{Lemma}
\newtheorem{lem}[thm]{Lemma}
\newtheorem{prop}[thm]{Proposition}
\newtheorem{proposition}[thm]{Proposition}

\newtheorem{definition}[thm]{Definition}
\newtheorem{cor}[thm]{Corollary}

\newtheorem {rem}[thm]{Remark}
\newtheorem {remark}[thm]{Remark}
\newtheorem {remarks}[thm]{Remarks}

\newtheorem*{hypos}{Hypotheses}

\newcommand{\bD}{ \mathbb{D}}
\newcommand{\bra}{\langle}
\newcommand{\ket}{\rangle}
\newcommand{\ep}{\hfill {$\Box$}}

\newcommand{\japAs}{\bra A \ket^{-s}}

\title{Limiting absorption principle for contractions
\thanks{Supported by FONDECYT 1211576, ECOS-ANID 200035 and {\it Concurso Movilidad de Profesores y Apoyo a Visitantes Extranjeros}, Facultad de Matem\'aticas, PUC}
}
\author{Joachim Asch \thanks{Universit\'e de Toulon, CNRS, CPT, UMR 7332, 83957 La Garde, France \\ Aix-Marseille Univ, CNRS, CPT, UMR 7332, Case 907, 13288 Marseille, France, asch@cpt.univ-mrs.fr},
Olivier Bourget
\thanks{
Facultad de Matem\'aticas,
Pontificia Universidad Cat\'olica de Chile, Av. Vicu\~{n}a Mackenna 4860,
C.P. 690 44 11, Santiago, Chile, bourget@uc.cl}
}
\date{18/5/24}
\begin{document}
\maketitle

\begin{abstract}
We establish limiting absorption principles for contractions on a Hilbert space. Our sufficient conditions are based on positive commutator estimates. We discuss the dynamical implications of this principle to the correspon\-ding discrete-time semigroup and provide several applications. Notably to Toeplitz operators and contractive quantum walks.
 \end{abstract}

%%%%%%%%%%%%%%%%%%%%
%%%%%%%%%%%%%%%%%%%%
\section{Introduction}

The limiting absorption principle states that there exists a topology in which the resolvent can be continuously extended to parts of the essential spectrum. It was originally developed for resolvents of selfadjoint Schr\"odinger operators \cite{agmon} and is widely used to establish propagation properties of the associated strongly continuous group  and perturbations thereof, see \cite{yaf}.

 In particular it provides information on the absolutely continuous subspace and plays an important role in the proof of asymptotic completeness of the quantum mechanical $N-$ body problem and in the  study of the dynamics of embedded eigenvalues.

In the present contribution  we are interested in the dynamics of non-isolated systems  modeled by a discrete semigroup $\left( V^n\right)_{n\in\bN}$ for a contraction $V$. 

We consider
$\mathcal{H}$  a separable Hilbert space,  its bounded operators  $\cB(\mathcal{H})$ and a contraction
\[V\in\cB(\mathcal{H}), \quad \Vert V\Vert=1.\]

{ $1\in\cB(\hil)$ denotes the identity operator}. By a limiting absorption principle  for $V$  with respect to a weight $W\in\cB(\mathcal{H})$ we mean that :

\[\left\{z\in\bC; \vert z\vert <1\right\}=:\bD\ni z\mapsto  W (1-zV^*)^{-1} W^\ast
\]
extends to a norm continuous function on a suitable subset $\bf{D}\subset\overline\bD$.

In particular 
\[\sup_{ z\in\bf{D}}\Vert  W(1-zV^*)^{-1}W^\ast\Vert < \infty.\]

For $\psi \in \overline{\ran W^\ast}$ this implies  square summable decay of the correlations $\bra \psi, V^n\psi\ket $, more precisely $\psi$ is an element of
\[
\hil_{ac}(V):=\overline{ \left\{\psi \in\hil: \exists C_{\psi} \geq 0, \sup_{\|\varphi\|=1}\sum_{n=0}^{\infty} \left| \left\bra \varphi,V^n \psi\right\ket \right|^2 
\leq C_{\psi}\right\}} ,
\]
see Proposition \ref{lapToHac}. $\hil_{ac}(V)$ is called the absolutely continuous subspace ; for unitary $V$, it is the space of vectors of absolutely continuous spectral measure.

\vskip1cm
While there exists a large body of literature concerning limiting absorption principles, let us just briefly mention some work which concerns the non selfadjoint case :

\smallskip

The absolutely continuous subspace for the  generator of a continuous contraction semigroup was introduced by Davies \cite{D78}, see also \cite{kato66}, and used in his non-unitary scattering theory, \cite{davies1} see \cite{fffs} 
for a recent development.

Limiting absorption principles in a non-selfadjoint settings has been developed by Royer \cite{R1,R2}, see also \cite{BG} for interesting information.

{Limiting absorption principles for unitary operators and the related propagation properties} for the corresponding discrete group have been established in \cite{ABCF}, \cite{frt}, \cite{abc2}, see also Kato \cite{katosmooth}.

The theory of characteristic functions can be  used to obtain complementary spectral information in the case of trace class perturbations of unitary operators, see \cite{liawtreil}.
\subsection{Main results}

We now state our conditions on the contraction $V$ and our results. For the proofs we use the positive commutator method pioneered by Mourre 
\cite{mourreOriginal}, see also \cite{abmg}, which makes use of an escape observable i.e.: an unbounded selfadjoint operator which we consistently call $A$ in the sequel. 

\begin{definition}
Let $A$ be a selfadjoint operator. For a  bounded operator $B$ denote ${Ad}_{e^{-i A t}}(B):= e^{-iAt}B e^{iAt}\,\, (t\in\bR)$. We say:
\begin{enumerate}
\item $B\in C^k(A)$ for $k\in\mathbb{N}$ if  ${Ad}_{e^{-i A t}}(B)$ is strongly $C^k$;
\item $B\in {\cal C}^{1,1}(A)$ if:
\begin{equation*}
\int_0^1 \|{Ad}_{e^{-i A t}}(B)+{Ad}_{e^{i A t}}(B)-2B\|\,\frac{d\tau}{|\tau|^2} < \infty \enspace .
\end{equation*}
\item{For $B\in C^1(A)$, define the commutator $\lbrack A, B\rbrack :=ad_A(B):= i\partial_t {Ad}_{e^{-i A t}}(B)\rvert_{t=0}$.}
\end{enumerate}
\end{definition}

\begin{remarks}
\begin{enumerate}
\item One has $C^2(A)\subset {\cal C}^{1,1}(A)\subset C^1(A)$ and ${\cal C}^{1,1}(A)$ is a $*$-algebra, see \cite{abmg}.
\item If $U\in C^1(A)$ is unitary then  $U^*AU-A$ extends from the domain  of $A$ to $\bB(\hil)$ and  $U^*AU-A=U^\ast ad_A(U)$, see \cite{abc2}.
\item For $U$ unitary $U\in{\cal C}^{1,1}(A)$ is the minimal regularity assumption needed for proving the limiting absorption principle to hold using Mourre methods \cite{abmg, abc2}. 
\end{enumerate}
\end{remarks}

Our first result is a global limiting absorption principle under the rather strong assumption of existence of a positive commutator. We use the notation $\bra A \ket:=(A^2+1)^{1/2}$.

\begin{theorem}\label{lap3} Let $V\in\cB(\hil)$, $\Vert V\Vert=1$. Assume: 
 $V\in C^{1,1}(A)$ and there exists an $a_0>0$ such that $\Re \left(V^\ast ad_A V \right) \ge a_0\bI$.
 
Then for  $s>\frac{1}{2}$ the map
$ \bD\ni z\mapsto  \bra A \ket^{-s}(1-zV^*)^{-1}\bra A \ket^{-s}
$
extends continuously in the uniform topology to $\overline\bD$; in particular

\begin{equation*}
\sup_{z \in \bD} \|\bra A \ket^{-s}(1-zV^*)^{-1}\bra A \ket^{-s} \| < \infty \hbox{ and }
\end{equation*}
\[\overline{\ran \japAs}\subset {\cal H}_{ac}(V)\cap {\cal H}_{ac}(V^\ast).
\]
\end{theorem}

\begin{remark}\label{remark:cnu}
Remark that a contraction can always be decomposed in a direct sum of its unitary and completely non-unitary parts \cite{nf}. Theorem \ref{lap3} and Theorem \ref{lap2} below are known to hold for a unitary $V_u$ whereas for a completely non-unitary $V_{cnu}$ it is known that ${\cal H}_{ac}(V_{cnu})\cap {\cal H}_{ac}(V_{cnu}^\ast)=\hil$, see Proposition \ref{cnu} below. This decomposition and information on the unitary part may be difficult to obtain in applications, we will not make use of it {in the present contribution}.

\end{remark}
For our second main result it is sufficient to assume positivity of the commutator  locally in the spectrum of a unitary reference operator which we always call $U$.

We denote $\bT:={\mathbb R}/2\pi {\mathbb Z}$ and use its identification with the unit circle $\partial\bD=\exp(i \bT)$ .
Denote the spectral family of $U$ by $E(\Theta)$, where $\Theta$ belongs to the Borel sets of $\bT$. Thus for bounded Borel functions $\Phi$ on $\partial\bD$, we have
\[\Phi(U)=\int_{\mathbb T}\Phi(e^{i\theta})dE(\theta) .
\]

\begin{definition}
Let $U\in C^1(A)$.  We say that:

 $U$ satisfies a Mourre estimate w.r.t. $A$ on the Borel set $\Theta$ if there exist $a >0$ and a compact operator $K$ such that

\begin{equation}\label{mourre-aK}
E(\Theta)(U A U^* -A)E(\Theta) \geq a E(\Theta) +K ,
\end{equation}
The  Mourre estimate is called {\emph strict} if $K=0$.
\end{definition}

\bigskip
{For an operator $B$ denote $|B| := \sqrt{B^*B}$. We assume 

\begin{hypos}[${\bf H}$] For  a  selfadjoint operator  $A$,
\begin{itemize}
\item[(H1):] there exists a unitary operator $U$, $U\in C^1(A)$  such that a Mourre estimate \eqref{mourre-aK} holds on an open subset $\Theta \subset {\mathbb T}$;
\item[(H2):] $V=PUQ$ for $P,Q$  such that $0\le P\le 1$, $0\le Q\le 1$;
\item[(H3):] $V\in C^{1,1}(A)$;
\item[(H4):] Let $a$ be the Mourre constant defined in (\ref{mourre-aK}). For $W=UV^\ast$ and for $W=U^\ast V$, it holds: there exists a compact selfadjoint operator $K_W$ such that
\begin{equation}\label{eq:h4c}
\Vert i ad_A (\Im(W)) -K_W \Vert<a
\end{equation}
and there exists $ \alpha >0$, such that
\begin{equation}{\label{eq:h4p}}
2 \Re (1-W) - (1+\alpha) |1-W |^2 \geq 0.
\end{equation}
\end{itemize}
\end{hypos}

\begin{remarks} Hypothesis (H4) is technical and used in Section \ref{prelim-est:L}. 
\begin{enumerate}
\item (H4) is not a restriction if in (H2) $P=1$ or $Q=1$ c.f. Remark \ref{farfromnd+weights} below.
\item (H4) is not a restriction if in (H2) $P=Q$ and $P$ is an orthogonal projection such that $[U,P]$ is compact, c.f. Remark \ref{p=q} below.
\end{enumerate}
\end{remarks}

Concerning  eigenvalues on the unit circle, we will prove 

\begin{proposition}\label{virial}  Assume $(H1), (H2)$. Then
\begin{enumerate} 
\item For $\mu \in \partial\bD, \psi\in\hil\setminus\{0\}$:  $V\psi=\mu \psi\Longrightarrow U\psi=\mu \psi$.
\item If a strict Mourre  estimate (\ref{mourre-aK}) holds with $K=0$ in $\Theta$, then $V$ has no eigenvalues in $e^{i\Theta}$.
\end{enumerate}
\end{proposition}

We now state our second main result, a local limiting absorption principle in the subset $\Theta$ of the spectrum of $U$ where the Mourre estimate (\ref{mourre-aK}) holds: 

We denote  by ${\mathcal E} (B)$ the set of eigenvalues of  an operator $B$.

\begin{theorem}\label{lap1} Assume $(H1-4)$. Then for any  $s>1/2$ the map

 $z\mapsto \bra A \ket^{-s}(1-zV^*)^{-1}\bra A \ket^{-s}$ 
 extends continuously from ${\mathbb D}$ to ${\mathbb D}\cup e^{i\Theta} \setminus\cE(U)$ in the operator norm topology; in particular
\begin{equation*}
\sup_{z \in [0,1)\cdot e^{i\Theta_0}} \|\bra A \ket^{-s}(1-zV^*)^{-1}\bra A \ket^{-s} \| < \infty 
\end{equation*}
for any closed set $\Theta_0$ such that $e^{i{\Theta_0}} \subset e^{i\Theta} \setminus \cE( U )$ 
\end{theorem}

Our third version of the limiting absorption principle is : 

\begin{theorem}\label{lap2}  Assume $(H1-4)$. Then for any  $s>1/2$ and any closed set $\Theta_0$ such that $e^{i{\Theta_0}} \subset e^{i\Theta} \setminus\cE( U )$, any $\Phi \in C^{\infty} (\partial\bD;{\mathbb R})$  supported on $e^{i{\Theta_0}}$:
\begin{gather*}
\sup_{z\in\bD} \|\bra A \ket^{-s} \Phi (U) (1-zV^*)^{-1} \Phi (U) \bra A \ket^{-s} \| < \infty \hbox{ and }
\end{gather*}
\[\overline{\mathrm{Ran} \, \Phi (U) \bra A \ket^{-s}} \subset {\mathcal H}_{ac}(V)\cap\hil_{ac}(V^\ast).\]
\end{theorem}

We will illustrate our abstract results in Section \ref{two} through various examples. In particular, we discuss the role played by the  hypotheses. The following sections are dedicated to the proofs.  In Section \ref{dynamics}, we prove Proposition \ref{virial} and relate the limiting absorption principle to the control of the absolutely continuous subspace as stated in Theorems \ref{lap3} and \ref{lap2}. The proof of the limiting absorption principles  is developed in Section \ref{prooflap1}. We start with some auxiliary results on Mourre inequalities in Section \ref{mourre:2}. The proof of our limiting absorption principle under the  ${\mathcal C}^{1,1}(A)$ regularity condition, which is optimal on the scale of $\cC^{s,p}(A)$ spaces  as developed in \cite{abmg}, requires the technical developments of Section \ref{prop:V}. In Section \ref{prelim-est}  we then proceed with some a priori estimates on a weighted version of a suitably  deformed resolvent. In Section \ref{diff:ineq}, we establish some differential inequalities on this weighted deformed resolvent, which, once combined with the a priori estimates, allows to conclude  the proofs of Theorems \ref{lap3} and \ref{lap1}. Some complementary aspects related to the proof of Theorem \ref{lap2} are developed in Section \ref{prooflap2}. 

%%%%%%%%%%%%%%%%%%%%%%%%%%%%%%%%%%%%%%%%
%%%%%%%%%%%%%%%%%%%%%%%%%%%%%%%%%%%%%%%%
\section{Applications and discussion of the hypotheses}\label{two}

We will use freely:
\begin{remark}
%\phantom{space}$B\in C^1(A)$ if and only if the sesquilinear form \phantom{another word}
$B\in C^1(A)$ if and only if  the sesquilinear form
$$
{\cal D}(A)\times {\cal D}(A)\ni (\varphi,\psi)\mapsto\langle A\varphi,B\psi\rangle - \langle \varphi,BA\psi\rangle
$$
extends continuously to ${\cal H}\times {\cal H}$. In this case, the bounded operator associated to its extension  is $\mathrm{ad}_A (B)$. 
\end{remark}

\subsection{Fundamental example}
We illustrate Theorem \ref{lap2} with the rather basic but instructive example of a specific rank 1 perturbation of the shift operator.  

With the normalized Lebesgue measure $d\ell$, let $\hil:=L^2(\partial\bD,d\ell)$, $U\psi(z):=z\psi(z)$, $Q\psi(z):=\psi(z)-\int_{\partial\bD}\psi d\ell$, $V:=UQ$.
Remark that
\[Vz^n=\left\{
\begin{array}{ll}
  z^{n+1}   & n\neq0  \\
  0  &  n=0 
\end{array}
\right.
\]
so $\ker(V)= span\{z^0\}$ and on $\ker{V}^\perp$ the contraction decouples to the forward shift and a unitary equivalent to the  backward shift. In particular
\[\hbox{{spectrum}}(V)=\overline\bD\quad \hbox{ and } \cE(V)={\bD}.\]
While one can see explicitly from the definition of $V$ that $\hil_{ac}(V)\cap\hil_{ac}(V^\ast)=\hil$, our hypothesis are satisfied and we can apply our theorem. Indeed, for

\[A:=z\partial_z \hbox{ on } D(A):=\left\{\sum_{n\in \bZ} a_nz^n; \sum_{n\in\bZ}(1+n^2)\vert a_n\vert^2 <\infty
\right\}\]
it holds
\[e^{-iAt}Ue^{iAt}=e^{-it}U, \quad e^{-iAt}Qe^{iAt}=Q, 
\]
in particular $U$ and $Q$ are in $C^\infty(A)$ and $(U^\ast A U- A)=1$ so a strict Mourre estimate (\ref{mourre-aK}) holds with constant $a=1$ on any measurable subset $\Theta$.

Concerning hypothesis (H4), remark that $V^\ast U=Q$, $UV^\ast=UQU^\ast$, thus for $W\in\{ V^\ast U, UV^\ast\}$: $\Im(W)=0$ so \eqref{eq:h4c} always holds with $K_W=0$. Also in both cases $0\le\Re(W)\le1$, thus $0\le1-\Re(W)\le1$, which implies $1-\Re(W)\ge \vert 1-\Re(W)\vert^2$ and \eqref{eq:h4p} holds with $\alpha=1$. 

So {$(H1-4)$} are satisfied and we can apply Theorem \ref{lap2} for $\Phi=id$ and conclude that for $s>\frac{1}{2}$, $\hil_s:=\left\{\sum_{n\in \bZ} a_nz^n; \sum_{n\in\bZ}{(1+n^2)}^s\vert a_n\vert^2 <\infty\right\}$:
\[\overline{\hil_s}=\hil=\hil_{ac}(V)\cap\hil_{ac}(V^\ast).\]

\begin{remark}\label{farfromnd+weights} 
It follows from this argument that in general if $P=1$, i.e. $V=UQ$, then Hypothesis $(H2)$ implies $(H4)$. The same is true if $V$ is of the form $V=PU$.
\end{remark}

\begin{remark} The unilateral forward shift {$V\psi(z):=z\psi(z)$} on the Hardy space
$$
\hil:=\left\{\sum_{n\in \bN} a_nz^n; \sum_{n\in\bN}\vert a_n\vert^2 <\infty
\right\}
$$
is an example for Theorem \ref{lap3}. \\
{Indeed,} with 
$A:=z\partial_z \hbox{ on } D(A):=\left\{\sum_{n\in \bN} a_nz^n; \sum_{n\in\bZ}(1+n^2)\vert a_n\vert^2 <\infty
\right\}$, it holds $V\in C^\infty(A)$ and
$V^\ast ad_A(V)=1$, 
so \ref{lap3} applies.
\end{remark}

\subsection{Contractive convolution operators, quantum walks}
In $\hil=\ell^2\left(\bZ^d; \bC^{d^\prime}\right)$ consider $V=UP$ for $U=C_0C_1$ , $P=P_0P_1P_0$  where $C_1,P_1$ are convolution operators and $P_0,C_0$  matrix valued multiplication operators with $C_j$ unitary {and $0\le P_j\le1$}. 

More specifically, denote $U(d^\prime)\subset \bM_{d^\prime, d^\prime}$ the unitary group  and $\cP_{01}(d^\prime)=\{Q\in \bM_{d^\prime, d^\prime}; 0\le Q\le 1\}$. For the symbols $f\in L^\infty(\bT^d,U(d^\prime))$, $g\in \ell^\infty(\bZ^d,U(d^\prime))$, and $p\in L^\infty(\bT^d,\cP_{01}(d^\prime))$,  $q\in \ell^\infty(\bZ^d,\cP_{01}(d^\prime))$ consider
\[C_1\psi(x)=\sum_{y\in\bZ}\widehat{f}(x-y)\psi(y), \quad C_0\psi(x)=g(x)\psi(x)\]
\[P_1\psi(x)=\sum_{y\in\bZ}\widehat{p}(x-y)\psi(y) , \quad P_0\psi(x)=q(x)\psi(x).\]

Suppose that 
\[{f} \hbox{ is analytic and  } \int_1^\infty\sup_{r\le\vert x\vert\le 2r}\Vert g(x)-1\Vert dr<\infty.\]
\[{p}\in C^3(\bT^d, \cP_{01}(d^\prime)) \hbox{ and } \int_1^\infty\sup_{r\le\vert x\vert\le 2r}\Vert q(x)-1\Vert dr<\infty.\]

Remark that contractive quantum walks with asymptotically periodic coins and local absorption are a particular example of the above \cite{HJ2,abj3}. 

We argue that hypotheses {$(H1-4)$} are satisfied:

\smallskip

Denote the selfadjoint operator $X\psi(x):=x\psi(x)$ defined on $D(X)=\{\psi\in\hil; \sum_x (1+x^2)\vert\psi(x)\vert^2<\infty\}$. Then there exists a selfadjoint propagation observable $A$ such that $U\in{\mathcal C}^{1,1}(A)$ and a discrete subset {$\tau_f\subset\bT$} such that a  Mourre estimate (\ref{mourre-aK}) holds on every open $\Theta$ such that $\overline{\Theta}\subset\bT\setminus {\tau_f}$. 

Furthermore $A$  is relatively bounded, $A^s\langle X\rangle^{-s}\in \cB(\hil), s\in \{1,2 \}$ and  $P\in {\mathcal C}^{1,1}(A)$, 
we again refer to \cite[section 3]{abj3} for proofs.   

Also  $0\le P\le1$; Hypotheses $(H1)-(H3)$ are satisfied and by Remark \ref{farfromnd+weights} also $(H4)$. Thus Theorem \ref{lap2} applies.

%%%%%%%%%%%%%%%%%%%%%%%%%%%%%%%%%%%%%%%%
%%%%%%%%%%%%%%%%%%%%%%%%%%%%%%%%%%%%%%%%

\subsection{Toeplitz operators}
Let $\hil=\ell^2\left(\bZ; \bC\right)$, $V=PUP$ with $P$  the multiplication by the characteristic function of the half line: $P\psi(x)=\chi(x\ge0)\psi(x)$, and  $U$ a unitary  convolution operator with symbol $f\in L^\infty(\bT,U(1))$, $U\psi(x):=\sum_{y\in\bZ}\widehat{f}(x-y)\psi(y).$

Remark that the restriction of $V$ to $\ran{P}$ is equivalent to the Toeplitz operator with symbol $f$.

We now show that Hypotheses $(H1)-(H4)$ are satisfied if $f$ is smooth enough and Theorem \ref{lap2} applies. 

If $f\in C^3(\bT,U(1))$ and {$ f^\prime \hbox{ has no zeros in } 
\Theta\subset\bT$},
then there exists a propagation observable $A$ such that $U\in {\mathcal C}^{1,1}(A)$ and such that a Mourre estimate (\ref{mourre-aK}) holds in $\Theta$, see \cite[section 3]{abj3} and \cite{abc3}. 

In order to prove smoothness of the Hardy projection $P$, we recall the construction of $A$.

Denote $L_g$ be the convolution operator with symbol $g$ {on $\hil$}. For $g:= if \bar{f'} $ the selfadjoint operator $A:= \frac{1}{2} \left( L_g X + XL_g \right)=L_gX+\frac{i}{2}L_{g^\prime}$ is defined by extension from $D(X)$ which is a core for $A$.

It holds

\begin{prop} If $f\in C^6(\bT,U(1))$ then $P\in {\mathcal C}^{1,1} (A)$.
\end{prop}

\begin{proof} 
We show first that $\langle X\rangle^\alpha\mathrm{ad}_A P $ is a bounded operator for $\alpha\in\{0,1\}$.

Let $P^\perp:=1-P$. For a symbol $h\in C^4(\bT, \bC)$, we can estimate the matrix elements of $P^{\perp}  L_h  P$ on the canonical  basis of $l^2 ({\mathbb Z})$:
\[
\begin{split}
| \bra e_k, P^{\perp} L_h Pe_l \ket |& =\chi(k<0)\chi(l\ge0) \left( |\bra e_k, L_h e_l \ket | = |\hat{h}_{l+|k|} | \right)\\
&\le \frac{const}{ (|k|+|l|)^4} \leq \frac{const}{ \bra k\ket^{2} \bra l\ket^{2 }}
\end{split}
\]
For $\alpha,\beta\in\{0,1\}$ it follows 
\[| \bra e_k, X^\alpha P^{\perp} L_h P X^\beta e_l \ket | \leq \frac{const \bra k\ket^{\alpha} \bra l\ket^{\beta }}{ \bra k\ket^{2} \bra l\ket^{2 }}
\]
which implies that the Hilbert-Schmidt norm of $ X^\alpha P^{\perp} L_h P X^\beta $ is finite. {It holds:
\[P^{\perp} A P=P^{\perp}L_gPX+\frac{i}{2}P^{\perp}L_{g^\prime} P\]
with $g, g'\in C^4(\bT,\bC)$}. It follows that $ad_A P=P^\perp A P-PAP^\perp$ and $X ad_A P$ are bounded operators.

For $\chi \in C^{\infty} ({\mathbb R}; [0,\infty))$ supported on an interval $(a,b)$, $0<a<b$ one has :
\begin{align*}
\int_1^{\infty} \| \chi(\bra X\ket/r) \mathrm{ad}_A P \| \frac{dr}{r} &\leq \int_1^{\infty} \| \chi(\bra X\ket/r) \bra X\ket^{-1} \| \| \bra X\ket \mathrm{ad}_A P \| \frac{dr}{r} \\
&\leq \| \bra X\ket \mathrm{ad}_A P \| \int_1^{\infty} \frac{dr}{ar^2} < \infty .
\end{align*}
In view of the criterion provided by Theorem 7.5.8 \cite{abmg} (also Theorem 3.2 in \cite{abc3}), we deduce that $P \in {\mathcal C}^{1,1} (A)$.
\end{proof}

It is a corollary that $\lbrack U,P\rbrack$ is Hilbert Schmidt for $f\in C^4(\bT,\bC)$ and compact for continuous $f$.

\begin{remark}\label{p=q} Concerning hypothesis (H4). We have $UV^\ast = UPU^\ast P$ , so $2i\Im UV^\ast = P^\perp [U,P]U^\ast P - PU[P,U^\ast] P^\perp$. 

 $[U,P]$ compact and $U, P\in {\mathcal C}^{1,1}(A)$ so \eqref{eq:h4c} holds for $W=UV^\ast $with $K_W=i\mathrm{ad}_A \Im W$ for any $a>0$. On the other hand,
\begin{align*}
2 \Re (1-UV^\ast) &= 2(1-PUPU^\ast P) -P^\perp UPU^\ast P - PUPU^\ast P^\perp \\
|1-UV^\ast |^2 &= 1-PUPU^\ast P -P^\perp UPU^\ast P - PUPU^\ast P^\perp  
\end{align*}
Now observe
\begin{align*}
2 \Re (1-UV^\ast) - \frac{3}{2} |1-UV^\ast |^2 &= \frac{1}{2} (1-PUPU^\ast P) + \frac{1}{2} (P^\perp UPU^\ast P + PUPU^\ast P^\perp) \\
&= \frac{1}{2} (P^\perp + PUP^\perp U^\ast P) - \frac{1}{2} (P^\perp UP^\perp U^\ast P + PUP^\perp U^\ast P^\perp) \\
&= \frac{1}{2} | P^\perp - UP^\perp U^\ast P |^2 .
\end{align*}
which is positive so \eqref{eq:h4p} holds with $\alpha=1/2$.
Similarly one shows that \eqref{eq:h4c} and \eqref{eq:h4p} hold for $W=U^\ast V= U^\ast PUP$ and we conclude that the hypotheses are satisfied and we can apply our theorem.
\end{remark}

%%%%%%%%%%%%%%%%%%%%%%%%%
%%%%%%%%%%%%%%%%%%%%%%%%%

\section {Dynamics}\label{dynamics}
To prove the dynamical implications we extend the techniques known for unitary operators.

\subsection{Absolutely continuous subspace}

We prove that the limiting absorption principle implies that certain vectors belong to $\hil_{ac}$ which proves the corresponding assertions in Theorems \ref{lap3} and \ref{lap2}.

\begin{proposition}\label{lapToHac}
For a contraction $V$ and $W\in\cB(\hil)$ it holds: if
\[\sup_{ z\in\bD}\Vert  W(1-zV^*)^{-1}W^\ast\Vert < \infty ,\]
then
\[ \overline{Ran W^\ast}\subset \hil_{ac}(V)\cap \hil_{ac}(V^\ast) .
\]
\end{proposition}

\begin{proof} The result is known for the case where $V$ is unitary, \cite[Theorem 2]{ABCF}.
Consider the operator
\[P_z(V):=2\Re\left((1-zV^\ast)^{-1}\right)-1=(1-zV^\ast)^{-1}\left(1-\vert z\vert^2V^\ast V \right)(1-zV^\ast)^{-1\ast}.
\]
The assumption implies 
\[{c :=} \sup_{ z\in\bD}\Vert  WP_z(V)W^\ast\Vert < \infty.\] 
Let $\widehat{V}$ be a unitary dilation on a Hilbert space $\widehat{\hil}=\hil\oplus\hil^\perp$ as described in \cite{nf}. For the orthogonal projection $P_\hil$ {onto} $\hil$, it then holds for all $\psi\in\hil$:

$P_\hil \left(\widehat{V}^n\psi\oplus0\right)=V^n\psi$ and 
$P_\hil \left(\widehat{V}^{\ast n}\psi\oplus0\right)=V^{\ast n}\psi$

\noindent and thus $P_\hil \left(P_z(\widehat{V})\psi\oplus0\right)=P_z(V)\psi$. Define ${\bf W}:=W\oplus0\in\cB(\widehat\hil)$. Then, for $\widehat\varphi, \widehat \psi\in\widehat\hil,$

\[\left\vert
 \left\bra P_z(\widehat{V}){\bf W^\ast}\widehat\varphi, {\bf W^\ast}\widehat\psi\right\ket_{\widehat{\hil}} \right\vert=
 \left\vert \left\bra P_z({V}){W^\ast}\varphi, {W^\ast}\psi\right\ket\right\vert\le c\Vert\varphi\Vert\Vert\psi\Vert\le c\Vert\widehat\varphi\Vert_{\widehat{\hil}}\Vert\widehat\psi\Vert_{\widehat{\hil}}
\]
so
\[\sup_{ z\in\bD}\Vert  {\bf W}P_z(V){\bf W}^\ast\Vert=c.
\]
Now by \cite[Theorem 2]{ABCF}

\[\sum_{n\in\bZ}\Vert{\bf {\bf W}}{\widehat{V}}^n \widehat\varphi\Vert_\hil^2\le c \Vert\widehat\varphi\Vert_\hil^2.
\]
Taking  $\widehat\varphi=\varphi\oplus0$ for any $\varphi\in\hil$, it follows
\[
\sum_{n=0}^\infty\Vert W{V^\ast}^n\varphi\Vert^2+\sum_{n=0}^\infty\Vert WV^n\varphi\Vert^2\le c\Vert\varphi\Vert^2.
\]
Now for $\psi=W^\ast\eta$
\[\left\vert\bra\varphi, V^n\psi\ket\right\vert^2=\left\vert\bra W{V^\ast}^n\varphi, \eta\ket\right\vert^2\le\Vert W{V^\ast}^n\varphi\Vert^2\Vert\eta\Vert^2 ,
\]
which implies $\psi\in\cH_{ac}(V)$. Repeating the argument with $V^\ast$ replacing $V$, we conclude {that} $\ran(W^\ast)\subset \hil_{ac}(V)\cap \hil_{ac}(V^\ast)$, which finishes the proof as the latter set is closed.
\end{proof}

The result on completely non-unitary contractions mentioned in Remark \ref{remark:cnu} follows from the well known result of Nagy-Foias \cite{nf4} by the same reasoning :
\begin{proposition}\label{cnu} Let $V$ be completely non-unitary. Then ${\cal H}_{ac}(V)\cap {\cal H}_{ac}(V^\ast)=\hil$.
\end{proposition}
\begin{proof} 
Let $\widehat{V}$ be the minimal  unitary dilation on a Hilbert space $\widehat{\hil}=\hil\oplus\hil^\perp$ as described in \cite{nf}. For the orthogonal projection $P_\hil$ {onto} $\hil$, it then holds for all $\psi \in \hil$:
 $P_\hil \left(\widehat{V}^n\psi\oplus0\right)=V^n\psi$ and 
 $P_\hil \left(\widehat{V}^{\ast n}\psi\oplus0\right)=V^{\ast n}\psi$. The spectrum of $\widehat{V}$ is absolutely continuous, see \cite{nf4} ,  which is equivalent to 
$\widehat\hil=\hil_{ac}(\widehat{V})$.

For $\widehat\varphi=\varphi\oplus0$ $\widehat\psi=\psi\oplus0$  one has
\[\sum_{n\in\bZ} \left| \left\bra \widehat\varphi,\widehat{V}^n\widehat\psi\right\ket \right|^2 
=\sum_{n\ge0} \left| \left\bra \varphi,{V}^n\psi\right\ket \right|^2 +\sum_{n>0} \left| \left\bra\varphi,{V^\ast}^n\psi\right\ket \right|^2 
\]
from which the assertion follows.
\end{proof}

%%%%%%%%%%%%%%%%%%%%%%%%%%%%%%%%%%%%%%%%
%%%%%%%%%%%%%%%%%%%%%%%%%%%%%%%%%%%%%%%%
\subsection{Eigenvalues, proof of Proposition \ref{virial}}\label{evcontrol}

\begin{lemma}\label{evAB} For two bounded operators $A, B$ suppose that  $\|A\| \leq 1$, $-1\leq B \leq 1$ and $\ker(1+B)=\{0\}$. Then for $\mu \in \partial {\mathbb D}, \varphi\in\hil ,$
\[ AB\varphi=\mu\varphi \Rightarrow  A\varphi =\mu \varphi \hbox{ and } B\varphi =\varphi.
\]
\end{lemma}
\begin{proof} First observe that: 
\[0= \|\varphi \|^2 - \|AB\varphi \|^2 =  \bra B\varphi, (1- A^\ast A) B\varphi \ket+\bra \varphi, (1-B^2) \varphi \ket .
\]
Both terms on the right hand side are non negative. It follows $\varphi\in\ker(\sqrt{1-B^2})$, which implies
\[B\varphi=\varphi \hbox{ and } AB\varphi=A\varphi=\mu\varphi \hbox{ as } -1\notin\cE(B).
\]
\end{proof}

Now, we prove Proposition \ref{virial}. 

\begin{proof} (of Proposition \ref{virial}) \begin{enumerate}
\item $V=PUQ$ with $U$  unitary and $0\le P\le1, 0\le Q\le1$. We can apply lemma~\ref{evAB} to $A:=PU$ and $B:=Q$ and we get for $\mu\in\partial\bD,$
\[V\psi=PUQ\psi = \mu \psi \Rightarrow  PU\psi=U(U^\ast PU)\psi=\mu\psi \hbox{ and } Q\psi=\psi.
\]
Observe that $0\le U^\ast PU\le1$ so we can apply Lemma \ref{evAB} again to $A:=U$ and $B:=U^\ast PU$ to conclude that $U\psi=\mu\psi$ and in addition that
\[ U^\ast PU\psi=\psi \hbox{ and thus } P\psi=\psi.\]

\item Suppose for $\mu\in e^{i\Theta}, V\psi=\mu \psi$ thus $U\phi=\mu\psi$ and 
\[\bra \psi, E(\Theta)(U^\ast A U-A) E(\Theta)\psi\ket=\bra \psi, (U^\ast A U-A)\psi\ket
\]

which was proven to  equal zero in \cite{abc2} Section 4.1. So $ E(\Theta)(U^\ast A U-A) E(\Theta)$ cannot be positive.
\end{enumerate}

\end{proof}

%%%%%%%%%%%%%%%%%%%%%%%%%%%%%%%%%%%%%%%%
%%%%%%%%%%%%%%%%%%%%%%%%%%%%%%%%%%%%%%%%
%%%%%%%%%%%%%%%%%%%%%%%%%%%%%%%%%%%%%%%%
%%%%%%%%%%%%%%%%%%%%%%%%%%%%%%%%%%%%%%%%

\section{Proof of Theorems \ref{lap3} and \ref{lap1}}\label{prooflap1}
We will prove Theorems \ref{lap3} and \ref{lap1} in parallel. We dub  \ref{lap3} the global case and \ref{lap1} the local case.

In the local  case of Theorem \ref{lap1}, we will show that the maps $z\mapsto \bra A\ket^{-s} (1-zV^*)^{-1} \bra A\ket^{-s}$ can be continuously extended to a neighborhood of any $e^{i\theta} \in e^{i\Theta}\setminus {\mathcal E} (U)$. In Section \ref{mourre:2}, we also have to rewrite the Mourre inequality in this neighborhood.

In Section \ref{prop:V}, we show how Hypothesis (H3) can be translated into the existence of suitable approximations for $V$ and $\mathrm{ad}_A V$. The construction of a deformed resolvent for $V$, denoted by $G_{\epsilon} (z)$, is based on this approximation. At this point, the proofs of Theorems \ref{lap3} and \ref{lap1} differ, see Sections \ref{prelim-est:G} and \ref{prelim-est:L} respectively. For convenience, we group the hypotheses as follows:

\begin{itemize}
\item (Glo) for the set of Hypotheses (H1') and (H3),

where (H1') stands for: there exists $a_0>0$ such that $\Re \left(V^\ast ad_A V\right) \ge a_0\bI$.
\item (Loc) for the set of Hypotheses (H1), (H2), (H3) and (H4).
\end{itemize}
In both cases, we establish some a priori estimates on the deformed resolvent and on weighted versions of this deformed resolvent, denoted by $F_{s,\epsilon} (z)$, {see \eqref{Fs-eps}}.

Next, we develop {Mourre's} differential inequality strategy in Section \ref{diff:ineq}. We deduce the continuous extension of the weighted resolvent at any point of $\partial {\mathbb D}$ under Assumption (Glo) and at any points $e^{i\theta} \in e^{i\Theta}\setminus {\mathcal E} (U)$ under Assumption (Loc) (Proposition \ref{mourre1}).

For any $S_1 \subset [0,\infty)$ and any $S_2 \subset {\mathbb T}$, we write
$$
S_1 \cdot e^{iS_2} = \{z\in {\mathbb C} , |z| \in S_1 , \arg z \in S_2 \} .
$$

%%%%%%%%%%%%%%%%%%%%%%%%%%%%%%%%%%%%%%%%
\subsection{Reduction to a strict Mourre estimate}\label{mourre:2}

If a Mourre estimate holds in a {set} not containing any eigenvalues, then a strict Mourre estimate holds in an open neighborhood of each point of this set. This remains true if one adds a compact operator to the commutator.

More precisely, 
\begin{lem}\label{Theta2Theta'} Suppose (H1). Let $e^{i\theta} \in e^{i\Theta}\setminus \cE(U)$ , $0 < c_1 < a$ and {${\bf K}$ a compact selfadjoint operator} then there exists an open connected neighborhood $\Theta'$ of $\theta$ such that:
\begin{align*}
E(\Theta' ) (U^*AU-A + {\bf K}) E(\Theta' ) &\geq c_1 E(\Theta' ) \\
E(\Theta' ) (A-UAU^* + {\bf K}) E(\Theta' ) &\geq c_1 E(\Theta' ) .
\end{align*}
\end{lem}
\begin{proof} Since $e^{i\theta} \in e^{i\Theta}\setminus \cE(U)$ and $K, {\bf K}$ are compact operators, we may find an open connected neighborhood $\Theta'$ containing $\theta$ such that $\| E(\Theta' ) (K+{\bf K}) E(\Theta' ) \| \leq a- c_1$.
\end{proof}

\begin{prop}\label{M0} Suppose (H1), (H2) and (H4). Let $e^{i\theta} \in e^{i\Theta}\setminus \cE(U)$.
There exist $0 < a_0 < a$, $a_1 >0$ and an open connected neighborhood $\Theta'$ of $\theta$ such that:
\begin{equation*}
\begin{split}
(U^*AU-A) - i\mathrm{ad}_A \Im UV^* &\geq a_0 - a_1 E(\Theta' )^{\perp} \\
(A-UAU^*) - i\mathrm{ad}_A \Im V^*U &\geq a_0 - a_1 E(\Theta' )^{\perp} .
\end{split}
\end{equation*}
\end{prop}
\begin{proof} Let $m < c_2 < c_1 < a$. where $$m := \max \left\{ \|i\mathrm{ad}_A \Im (UV^*)-K_{UV^*} \|, \| i\mathrm{ad}_A \Im (V^*U)-K_{V^*U} \| \right\}. $$ First, we apply Lemma \ref{Theta2Theta'} and fix the neighborhood $\Theta^\prime$ accordingly. Next, in view of Proposition \ref{mourre-equiv}, we deduce there exists $a_1 >0$ such that:
\begin{align*}
(U^*AU-A) + K_{UV^*} \geq c_2 - a_1 E(\Theta' )^{\perp} .
\end{align*}
Writing
\[
(U^*AU-A) - i\mathrm{ad}_A \Im UV^* = (U^*AU-A) - (i\mathrm{ad}_A \Im UV^* -K_{UV^*}) +K_{UV^*} 
\]
the first inequality follows with $a_0 = c_2 - m$ and the second equality analogously.

\end{proof}

%%%%%%%%%%%%%%%%%%%%%%%%%%%%%%%%%%%%%%%%
\subsection{Properties of $V$}\label{prop:V}

Now, we enumerate the properties which are implied by our hypothesis $V\in {\mathcal C}^{1,1}(A)$ and which are used in the proofs of Theorems \ref{lap3} and \ref{lap1}. We refer to \cite[Lemma 7.3.6]{abmg} and \cite{BST} for a proof of the following results which are used in the proof the  priori estimates and in the differential inequality procedure.

\begin{prop}\label{c11approx} 
Let $V \in {\mathcal C}^{1,1}(A)$. Then, there exists a map $S \in C^1 ((0,1); C^1 (A))$ (equipped with the operator norm topology), such that for $B := \mathrm{ad}_A S$, \\ $B \in C^1 ((0,1); C^1 (A))$  and :
\begin{itemize}
\item[] $\lim_{\epsilon \rightarrow 0} \epsilon^{-1}\| S_{\epsilon} - V\| =0$,
\item[] $\lim_{\epsilon \rightarrow 0} \| B_{\epsilon} - \mathrm{ad}_A V \|=0$,
\item[] $\sup_{\epsilon \in (0,\epsilon_0)} \| B_{\epsilon} \| <\infty$ for some $\epsilon_0 \in (0,1)$.
\end{itemize}
In addition,
\begin{align*}
\int_0^1 \frac{\| \partial_{\epsilon} S_{\epsilon} \|}{\epsilon} \, d\epsilon + \int_0^1 \big\| \mathrm{ad}_A B_{\epsilon} \big\| \, d\epsilon + \int_0^1 \big\| \partial_{\epsilon} B_{\epsilon} \big\| \, d\epsilon < \infty .
\end{align*}
\end{prop}

\begin{rem} In particular, there exists $C>0$ so that 
$
\| S_{\epsilon} - V\| \leq C\epsilon
$
for any $\epsilon \in (0,1)$ ; also the functions $S$, $B$ and $\partial_{\epsilon} S$ extend continuously  to  $[0,1)$
by setting $S_0=V$, $B_0=\mathrm{ad}_A V$, so that $(\partial_{\epsilon} S)(0) = 0$.
\end{rem}

\begin{cor}\label{Q-integrability} Let $V \in {\mathcal C}^{1,1}(A)$. For $\epsilon\in(0,1)$ define the map
\begin{equation}\label{Qeps}
{\mathcal Q(\epsilon)} = \partial_{\epsilon} S_{\epsilon} - \epsilon \partial_{\epsilon} B_{\epsilon} - \epsilon \mathrm{ad}_A B_{\epsilon} .
\end{equation}
Then,
$$
\epsilon \mapsto \frac{\| {\mathcal Q} (\epsilon )\|}{\epsilon} \in L^1 ((0,1)) .
$$
\end{cor}

\begin{rem} If $V\in C^2 (A)$, we can define for all $\epsilon$: $S_{\epsilon} \equiv  V$ and $B_{\epsilon} \equiv \mathrm{ad}_A V$.
\end{rem}

%%%%%%%%%%%%%%%%%%%%%%%%%%%%%%%%%%%
\subsection{Deformed resolvents and first estimates}\label{prelim-est}

In this section, we define the deformed resolvent $(1-zV^*)^{-1}$, then establish its invertibility and finally prove some a priori estimates, see Proposition \ref{firstbounds} below.

In view of Section \ref{prop:V}, we use the following shortcuts :
\begin{equation}\label{qQe}
\begin{split}
q_{\epsilon} & := \frac{S_{\epsilon}^* -V^*}{\epsilon} -  B_{\epsilon}^* - \mathrm{ad}_A V^* ,\\
Q_{\epsilon} & := \frac{S_{\epsilon}^* -V^*}{\epsilon} - B_{\epsilon}^* , 
\end{split}
\end{equation}
and $q_0 :=0$, $Q_0 := \mathrm{ad}_A V^*$. Note that $\lim_{\epsilon \rightarrow 0^+} \| q_{\epsilon} \| =0$.

For $\epsilon\in [0,\epsilon_0) , z\in \overline{\mathbb D}$, define
\begin{equation}\label{Teps(z)}
\begin{split}
V_\epsilon &:= S_{\epsilon} - \epsilon B_{\epsilon}\\
T_{\epsilon} (z) & := 1- zV_\epsilon^\ast \\
&=  T_0 (z) - z \epsilon Q_{\epsilon},
\end{split}
\end{equation}

To sum up:
\begin{lem}\label{Tepsilon-T0} Suppose (H3), then with 
\begin{equation}\label{b}
b := \sup_{\epsilon \in [0,\epsilon_0)} \| Q_{\epsilon} \| < \infty, 
\end{equation}
for any $\epsilon \in [0,\epsilon_0)$ and any $z\in \overline{\mathbb D}$ it holds
\begin{equation}\label{Teps-T0}
\big\| T_{\epsilon} (z) - T_0 (z) \big\| \leq b \epsilon .
\end{equation}
\end{lem}

%%%%%%%%%%%%%%%%%%%%%%%%%
%%%%%%%%%%%%%%%%%%%%%%%%%
\subsubsection{First estimates supposing (Glo)}\label{prelim-est:G}

\begin{prop}\label{mmtrick3} Suppose (Glo). There exists $0< \epsilon_1 < \epsilon_0$ such that for any $\epsilon \in [0,\epsilon_1) , z\in \overline{\mathbb D}\setminus\{0\},$
\begin{eqnarray}
T_{\epsilon}(z)^* + \bar{z} V_{\epsilon} T_{\epsilon}(z) &\geq & d(\epsilon, z) \label{mmt}\\
T_{\epsilon}(z) + z V_{\epsilon}^* T_{\epsilon}(z)^* &\geq & d(\epsilon, z) , \label{mmt*}
\end{eqnarray}
where
\begin{equation}\label{d}
d(\epsilon, z):= 1-|z|^2 + a_0 \epsilon |z|^2 .
\end{equation}
\end{prop}
\begin{proof} Fix $z\in \overline{\mathbb D}\setminus\{0\}$. Using the Mourre estimate and the contraction property for $V$, we have:
\begin{align*}
T_{\epsilon}(z)^* &+ \bar{z} V_{\epsilon} T_{\epsilon}(z) = 1- |z|^2 V_{\epsilon} V_{\epsilon}^* = 1-|z|^2 \left( V + (V_{\epsilon} -V)\right) \left( V^* + (V_{\epsilon} -V)^* \right) \\
&= 1-|z|^2 |V^*|^2 - 2|z|^2 \Re (V (V_{\epsilon}-V)^*) -|z|^2 |V_{\epsilon}^*- V^*|^2 \\
&\geq 1-|z|^2 - 2|z|^2 \Re (V (V_{\epsilon} -V)^* ) -|z|^2 |V_{\epsilon}^*- V^*|^2
\end{align*}
Mind that: $(V_{\epsilon}-V)^* = \epsilon \, ( \mathrm{ad}_A V^* + q_{\epsilon}) = \epsilon \, Q_{\epsilon}$, so $|(V_{\epsilon}-V)^* |^2 = \epsilon^2 | Q_{\epsilon} |^2$ and 
\begin{align*}
\Re (V (V_{\epsilon}-V)^*) = \epsilon \Re (V \mathrm{ad}_A V^*) + \epsilon \Re (V q_{\epsilon}) .
\end{align*}
It follows :
\begin{align*}
T_{\epsilon}(z)^* + \bar{z} V_\epsilon T_{\epsilon}(z) &\geq 1-|z|^2 - 2|z|^2\epsilon \Re (V \mathrm{ad}_A V^*) - 2|z|^2\epsilon \Re (V q_{\epsilon}) -|z|^2 \epsilon^2 | Q_{\epsilon} |^2 .
\end{align*}
We observe that: $0\leq | Q_{\epsilon} |^2 \leq b^2$ and $-\| q_{\epsilon} \| \leq \Re (V q_{\epsilon}) \leq \| q_{\epsilon} \|$. This yields:
\begin{align*}
T_{\epsilon}(z)^* + \bar{z} V_\epsilon T_{\epsilon}(z) \geq d(|z|, \epsilon) +  |z|^2 \epsilon \left( a_0 - b^2 \epsilon -2 \| q_{\epsilon} \| \right) .
\end{align*}
Pick $0 < \epsilon_1 < \epsilon_0$ such that
\begin{gather}\label{cond4epsilon1}
b^2 \epsilon_1 + 2 \sup_{\epsilon \in [0,\epsilon_1]} \| q_{\epsilon} \| \leq a_0 ,
\end{gather}
and \eqref{mmt} follows. The proof of \eqref{mmt*} can be done analogously . 
\end{proof}

For $d$ and $\epsilon_1$ respectively defined in \eqref{d} and \eqref{cond4epsilon1}, we define:
\begin{equation}
\Omega_{\mathbb T} = \{ (\epsilon, z) \in [0,\epsilon_1]\times\overline{\mathbb D}\setminus\{0\} ; d (\epsilon, z) >0 \} \label{omega:T}
\end{equation}
We deduce that:
\begin{prop}\label{invertibility} Suppose (Glo). For $(\epsilon, z) \in \Omega_{\mathbb T}$, the operator $T_{\epsilon} (z)$ is boundedly invertible.
\end{prop}
\begin{proof} Fix $\epsilon$, $z$ as stated. Making explicit  \eqref{mmt} and \eqref{mmt*} in terms of quadratic forms shows that $T_{\epsilon} (z)$ and $T_{\epsilon} (z)^*$ are injective from ${\mathcal H}$ into itself and has closed range. Since
\begin{equation*}
\overline{\mathrm{Ran} \, T_{\epsilon} (z)} = ( \mathrm{Ker} \, T_{\epsilon} (z)^*)^{\perp},
\end{equation*}
we deduce that $T_{\epsilon} (z)$ and $T_{\epsilon} (z)^*$ are actually linear and bijective hence boundedly invertible by the Inverse Mapping Theorem.
\end{proof}

%%%%%%%%%%%%%%%%%%%%%%%%%
%%%%%%%%%%%%%%%%%%%%%%%%%
\subsubsection{First estimates supposing (Loc)}\label{prelim-est:L}

We now look for an analog of \eqref{mmt} and \eqref{mmt*} under Assumptions (Loc). 

In this {subsection}, we fix $\theta \in \Theta \setminus {{\mathcal E}}(U)$ and look at the local properties of the resolvent on some open subset $\Theta_0$, such that $\overline{\Theta_0} \subset \Theta'$, where the open set $\Theta'$ has been defined in Proposition \ref{M0}.

We introduce some  shortcuts:
\begin{equation}\label{overRL}
\begin{split}
R &= UV^* \\
L &= V^*U \\
\overline{R} &=1-R \\
\overline{L} &=1-L .
\end{split}
\end{equation}
We also set: $E= E (\Theta' )$.

\begin{lemma}\label{eperp} Suppose (H1), (H2), (H3). Fix an open neighborhood $\Theta_0$ of $\theta$, such that $\overline{\Theta_0} \subset \Theta'$ and denote $d_0 :=$ dist$(\overline{\Theta_0}, {\mathbb T}\setminus \Theta') >0$. Then, we have for any $\epsilon \in [0,\epsilon_0)$, $z\in (0,1]\cdot e^{i\overline{\Theta_0}},$
{\begin{eqnarray}
E^{\perp} &\leq & \frac{3\pi^2}{4 d_0^2 |z|} \left( |T_{\epsilon}(z)|^2 + |z|^2 |\overline{R} |^2 + b^2\, \epsilon^2 |z|^2 \right) , \label{Eperp:R} \\
E^{\perp} &\leq & \frac{3\pi^2}{4 d_0^2 |z|} \left( |T_{\epsilon}(z)^* |^2 + |z|^2 |\overline{L}^*|^2 + b^2\, \epsilon^2 |z|^2 \right) , \label{Eperp:L}
\end{eqnarray}
}
with $b$ defined by \eqref{b}.
\end{lemma}
\begin{proof} For $\epsilon \in [0,\epsilon_0)$, we use shortcuts
\begin{equation}\label{VRLeps}
\begin{split}
R_{\epsilon} &:= UV_{\epsilon}^* , \\
&= R + \epsilon \, UQ_{\epsilon} , \\
L_{\epsilon} &:= V_{\epsilon}^* U , \\
&= L + \epsilon \, Q_{\epsilon} U .
\end{split}
\end{equation}
For $\epsilon \in [0,\epsilon_0)$, $z \in (0,1]\cdot e^{i\overline{\Theta_0}}$, we have
\begin{align*}
E^{\perp} &= (1-zU^*)^{-1} E^{\perp} \left(T_{\epsilon}(z) -zU^* \overline{R} +zU^* (R_{\epsilon} -R)\right) \\
&= (1-zU^*)^{-1} E^{\perp} \left(T_{\epsilon}(z) -zU^* \overline{R} +z \epsilon Q_{\epsilon}\right) ,
\end{align*}
which entails the first estimates. We conclude analogously for the second one.
\end{proof}

In Lemma \ref{identities} and Proposition \ref{from2Re2C}, we show how the local estimates obtained in Proposition \ref{M0} can be used to derive some estimates on $\Re (V \mathrm{ad}_A V^*)$ and $\Re (V^* \mathrm{ad}_A V)$ respectively.
\begin{lemma}\label{identities} Suppose (H1), (H2). With
\begin{equation*}\label{CU}
\begin{split}
C &:= (A-UAU^*) \\
C_U &:= U^*CU = {(U^*AU-A)}
\end{split}
\end{equation*} 
it holds:
\begin{align}
2 \Re (V \mathrm{ad}_A V^*) &= -2 C - 2 \overline{R^*} C \overline{R} +4 \Re (\overline{R^*} C) + 2\Re (R^* \mathrm{ad}_A R) \label{identities-1} \\
2 \Re (V^* \mathrm{ad}_A V) &= 2 C_U + 2 \overline{L}\, C_U \overline{L^*} -4 \Re (C_U \overline{L^*}) + 2\Re (L\, \mathrm{ad}_A L^*) . \label{identities-2}
\end{align}
\end{lemma}
\begin{proof} Since $U$ and $V$ belong to $C^1(A)$, $R$ and $L$ also belong to $C^1(A)$. Also, $V\in C^1(A)$ iff $V^*\in C^1(A)$ and $\mathrm{ad}_A V^* = - (\mathrm{ad}_A V)^*$. So, $\mathrm{ad}_A V^* = \mathrm{ad}_A (U^* R) = U^* (\mathrm{ad}_A  R) + (\mathrm{ad}_A U^*) R$ and
\begin{align*}
2 \Re (V \mathrm{ad}_A V^*) &= 2\Re (R^* \mathrm{ad}_A R) - 2\Re (R^* CR) ,
\end{align*}
from which \eqref{identities-1} follows. Similarly, \eqref{identities-2} follows from $\mathrm{ad}_A V = \mathrm{ad}_A (U L^*) = U (\mathrm{ad}_A  L^*)  + (\mathrm{ad}_A U) L^*$ and
\begin{align*}
2 \Re (V^* \mathrm{ad}_A V) &= 2\Re (L\, \mathrm{ad}_A L^*) + 2\Re (L\, C_U L^*) .
\end{align*}
\end{proof}

\begin{remark}\label{AB} We will use several times the following estimates. Let $A, B\in {\mathcal B} ({\mathcal H})$. Then, {for any $p>0$ and any $\varphi \in {\mathcal H}$,
\begin{align}
|\bra \varphi, 2\Re (A^*B) \varphi \ket | &\leq \|B \| \left( \frac{\| |A| \varphi \|^2}{p} + p \| \varphi \|^2 \right) . \label{2ReAB}
\end{align}
We also have that
\begin{align}
|\bra \varphi, A^*BA \varphi \ket | &\leq \|B \| \| |A| \varphi \|^2 . \label{A*BA}
\end{align}
}
\end{remark}

\begin{prop}\label{from2Re2C} Suppose (Loc). There exists $c_0 >0$, such that:
{
\begin{align*}
- 2 \Re (V \mathrm{ad}_A V^*)  &\geq \frac{3 a_0}{2} - 2a_1 E^{\bot} - c_0 |\overline{R} |^2 \\
2 \Re (V^* \mathrm{ad}_A V) &\geq \frac{3 a_0}{2} - 2a_1 E^{\bot} - c_0 |\overline{L}^* |^2 .
\end{align*}
}
\end{prop}
\begin{proof} Once noted that $2\Re (\mathrm{ad}_A R) = 2i \mathrm{ad}_A \Im R$, Lemma \ref{identities} writes:
$$
- 2 \Re (V \mathrm{ad}_A V^*) = 2 (C - i \mathrm{ad}_A \Im R) + 2 \overline{R^*} C \overline{R} - 4 \Re (\overline{R^*} C) + 2\Re (\overline{R^*} (\mathrm{ad}_A R)) .
$$
Applying inequality \eqref{2ReAB} to the couples $(A,B)= (\overline{R}, C)$ and $(\overline{R}, \mathrm{ad}_A R)$ yields for any $p>0$,
\begin{align*}
- 2\Re (\overline{R^*} C) &\geq - \|C \| \left( p^{-1} | \overline{R} |^2  + p \right) \\
2 \Re (\overline{R^*} (\mathrm{ad}_A R)) &\geq - \|\mathrm{ad}_A R \| \left( p^{-1} | \overline{R} |^2 + p \right) 
\end{align*}
while: $\overline{R^*} C \overline{R} \geq - \|C \| |\overline{R} |^2$. Summing up, for any $p>0$,
\begin{align*}
- 2 \Re ((\mathrm{ad}_A V) V^*) \geq \, &2 (C - i(\mathrm{ad}_A \Im R)) - p (2\| C\| + \|\mathrm{ad}_A R \|) - F_R (p) |\overline{R} |^2 ,
\end{align*}
where $F_R (p) = (2 \|C\| + \|\mathrm{ad}_A R \|) p^{-1} + 2\|C \|$. Fix $p=p_R$ where $2 p_R (2\| C\| + \|\mathrm{ad}_A R \|) = a_0$ and apply Proposition \ref{M0}; we get the first estimate with $c_{0,R}=F_R(p_R)$ instead of $c_0$.

Back to Lemma \ref{identities}, we can derive the second estimate analogously by defining the function $F_L$, the positive number $p_L$ and $c_{0,L}=F_L (p_L)$. We conclude by setting $c_0 := \max \{ c_{0,L}, c_{0,R} \}$.
\end{proof}

\begin{prop}\label{mmtrick} Suppose (Loc). There exists $0< \epsilon_1 < \epsilon_0$ such that for any $\epsilon \in [0,\epsilon_1]$, $z\in (0,1]\cdot e^{i\overline{\Theta_0}},$
\begin{eqnarray}
T_{\epsilon}(z)^* + \bar{z} V_{\epsilon} T_{\epsilon}(z) + \frac{3\pi^2 a_1}{2 d_0^2} \epsilon |z| |T_{\epsilon}(z)|^2 &\geq & d(\epsilon, z) , \label{mmt0} \\
T_{\epsilon}(z) + z V_{\epsilon}^* T_{\epsilon}(z)^* + \frac{3\pi^2 a_1}{2 d_0^2} \epsilon |z| |T_{\epsilon}(z)^*|^2 &\geq & d(\epsilon, z) , \label{mmt*0}
\end{eqnarray}
where
\begin{equation}\label{d0}
d(\epsilon, z) = 1-|z|^2 + a_0 \epsilon |z|^2 .
\end{equation}
\end{prop}
\begin{proof} Fix $z\in (0,1]\cdot e^{i\overline{\Theta_0}}$. We have:
\begin{align}
T_{\epsilon}(z)^* &+ \bar{z} V_{\epsilon} T_{\epsilon}(z) \nonumber \\
&= 1- |z|^2 R_{\epsilon}^* R_{\epsilon} = 1-|z|^2 \left( R^* + (R_{\epsilon} -R)^*\right) \left( R + (R_{\epsilon} -R) \right) \nonumber \\
&= 1-|z|^2 |R|^2 - 2|z|^2 \Re (R^* (R_{\epsilon}-R)) -|z|^2 |R_{\epsilon}-R|^2 . \label{mmt0:0}
\end{align}
Mind \eqref{qQe}. We have that: $R_{\epsilon} -R = \epsilon \, U ( \mathrm{ad}_A V^* + q_{\epsilon}) = \epsilon \, U Q_{\epsilon}$, so $|R_{\epsilon}-R |^2 = \epsilon^2 | Q_{\epsilon} |^2$ and 
\begin{align*}
\Re (R^* (R_{\epsilon}-R)) = \epsilon \Re (V \mathrm{ad}_A V^*) + \epsilon \Re (V^* q_{\epsilon}) .
\end{align*}
Rewriting $|R|^2 = |1-\overline{R}|^2 = 1 - 2\Re \overline{R} + |\overline{R}|^2$, \eqref{mmt0:0} yields
\begin{align*}
T_{\epsilon}(z)^* + \bar{z} V_{\epsilon} T_{\epsilon}(z) &= 1-|z|^2 + |z|^2 X_R + |z|^2 \epsilon Y - |z|^2 \epsilon\, Z_{\epsilon} \\
\text{where } \quad X_R &= 2 \Re \overline{R} - |\overline{R} |^2 \\
Y &= - 2 \Re (V \mathrm{ad}_A V^*) \\
Z_{\epsilon} &= 2 \Re (V^* q_{\epsilon}) + \epsilon | Q_{\epsilon} |^2 .
\end{align*}
We estimate the term $Y$ from below. Combining Proposition \ref{from2Re2C} with \eqref{Eperp:R} allows us to derive for $\epsilon \in [0,\epsilon_0)$, $z\in (0,1]\cdot e^{i\overline{\Theta_0}},$
\begin{equation}\label{be(z)}
\begin{split}
b_{\epsilon} + \frac{3 \pi^2 a_1}{2 d_0^2 |z|} |T_{\epsilon}(z)|^2 &\geq \left( \frac{3 a_0}{2} - \frac{3 \pi^2 a_1}{2 d_0^2} |z| b^2 \epsilon \right)  - |\overline{R} |^2 \left( c_0 + \frac{3 \pi^2 a_1}{2 d_0^2} |z| \right) .
\end{split}
\end{equation}
Taking advantage of the fact that $|z|\leq 1$, we get
\begin{equation*}
\begin{split}
T_{\epsilon}(z)^* &+ \bar{z} V_{\epsilon} T_{\epsilon}(z) + \frac{3\pi^2 a_1}{2 d_0^2} \epsilon |z| |T_{\epsilon}(z)|^2 \geq \\
&d(\epsilon, z) +  |z|^2 \epsilon \left( \frac{a_0}{2} - c_{\epsilon}  - \frac{3\pi^2 a_1}{2 d_0^2} b^2 \epsilon^2 \right) + |z|^2 \mathfrak{F}_{\epsilon}(\overline{R} )
\end{split}
\end{equation*}
where for $\epsilon \geq 0$
\begin{align}
\mathfrak{F}_{\epsilon}(\overline{R}) &:= 2 \Re \overline{R} - (1+\gamma_{\epsilon}) |\overline{R} |^2 , \label{PreQre} \\
\text{with } \quad \gamma_{\epsilon} &= \epsilon \left[ c_0 + \frac{3\pi^2 a_1}{2 d_0^2} \right]. \label{gamma:eps}
\end{align}
Now, we observe that: $0\leq | Q_{\epsilon} |^2 \leq b^2$ and $-\| q_{\epsilon} \| \leq \Re (V^* q_{\epsilon}) \leq \| q_{\epsilon} \|$. So, $\| c_{\epsilon} \|= o (\epsilon )$. This allows us to pick $0 < \mu < \epsilon_0$ such that 
\begin{align}
b^2 \mu^2 \frac{3\pi^2 a_1}{2 d_0^2} + \sup_{\epsilon \in [0,\mu]} \| c_{\epsilon} \| & \leq \frac{a_0}{2} , \label{cond4epsilon10} \\
\text{and }\quad \gamma_{\mu} & \leq \alpha . \label{gamma:eps1}
\end{align}
From \eqref{Te(z)*+zVTe(z)}, we obtain for all $\epsilon \in [0,\mu]$ and $z\in (0,1]\cdot e^{i\overline{\Theta_0}},$
\begin{equation}\label{Te(z)*+zVTe(z)}
T_{\epsilon}(z)^* + \bar{z} V_{\epsilon} T_{\epsilon}(z) + \frac{3\pi^2 a_1}{2 d_0^2} \epsilon |z| |T_{\epsilon}(z)|^2 \geq d(\epsilon, z) + |z|^2 ( 2 \Re \overline{R} - (1+\alpha) |\overline{R} |^2 ) .
\end{equation}
At this point of the proof, the condition \eqref{eq:h4p} of Hypothesis (H4), comes into scene to estimate the last terms on the RHS. This concludes the proof of \eqref{mmt0}. The proof of \eqref{mmt*0} is analogous. 

\end{proof}

For $d$ and $\epsilon_1$ respectively defined in \eqref{d0} and \eqref{cond4epsilon10}, let us write
\begin{align}
\Omega_0 := \{ (\epsilon, z) \in [0,\epsilon_1]\times (0,1]\cdot e^{i\overline{\Theta_0}} ; d (\epsilon, z) >0 \} . \label{omega0}
\end{align}
We deduce that:
\begin{prop}\label{invertibility:0} Assume (Loc). For $(\epsilon, z) \in \Omega_0$, the operator $T_{\epsilon} (z)$ is boundedly invertible.
\end{prop}
\begin{proof} Fix $\epsilon$, $z$ as stated. Reinterpreting \eqref{mmt0} (resp. \eqref{mmt*0}) in terms of quadratic forms shows that $T_{\epsilon} (z)$ (resp. $T_{\epsilon} (z)^*$) is injective from ${\mathcal H}$ into itself and has closed range. Since
\begin{equation*}
\overline{\mathrm{Ran} \, T_{\epsilon} (z)} = ( \mathrm{Ker} \, T_{\epsilon} (z)^*)^{\perp},
\end{equation*}
we deduce that $T_{\epsilon} (z)$ and $(T_{\epsilon} (z))^*$ are actually linear and bijective hence boundedly invertible by the Inverse Mapping Theorem.
\end{proof}

%%%%%%%%%%%%%%%%%%%%%%%%%
%%%%%%%%%%%%%%%%%%%%%%%%%
\subsubsection{Synthesis}

Aside from the local vs global aspects, the remaining components of the proofs of Theorems \ref{lap1} and \ref{lap3} are identical . From now, we unify notations with the introduction of 
\begin{equation}
\Omega := \left\{
\begin{array}{ll} \Omega_{\mathbb T} & \text{if (Glo) holds} \\
\Omega_0 & \text{if (Loc) holds}
\end{array} \right. \qquad \qquad {\bf S} := \left\{
\begin{array}{ll} \partial {\mathbb D}& \text{if (Glo) holds} \\
e^{i\overline{\Theta_0}} & \text{if (Loc) holds}
\end{array} \right. .
\end{equation}

Accordingly, assuming (Glo) or (Loc), Propositions \ref{invertibility} and \ref{invertibility:0} allow us to define for any $(\epsilon, z) \in \Omega,$ 
\begin{equation}
G_{\epsilon}(z) := (1-zV_{\epsilon}^*)^{-1} .\label{Ge}
\end{equation}

We have:
\begin{prop}\label{firstbounds} Suppose (Glo) or (Loc). Then
\begin{equation}\label{C0}
C_0:= \sup_{(\epsilon,z) \in \Omega} d (\epsilon, z) \| G_{\epsilon} (z) \| < \infty.
\end{equation}
In addition, for any $\epsilon \in (0,\epsilon_1]$, $z\in (0,1]\cdot {\bf S},$
\begin{align}
\max \left\{ \| G_{\epsilon} (z) \varphi \|, \| G_{\epsilon}^* (z) \varphi \| \right\} \leq \sqrt{\frac{2 |\Re \bra \varphi, G_{\epsilon} (z) \varphi \ket |}{a_0 \epsilon |z|^2}} \label{G-ReG}
\end{align}
\end{prop}

\begin{proof} For $(\epsilon, z) \in \Omega$, denote
\begin{equation}\label{He}
H_{\epsilon}(z) := \frac{1+zV_{\epsilon}^*}{1-zV_{\epsilon}^*}  = 2G_{\epsilon}(z) -1.
\end{equation}
We observe that
\begin{equation}\label{ReHe}
\begin{split}
\Re H_{\epsilon} (z) &= G_{\epsilon} (z)^* \left(T_{\epsilon}(z)^* + \bar{z} V_{\epsilon} T_{\epsilon}(z) \right) G_{\epsilon} (z) \\
&= G_{\epsilon} (z) \left(T_{\epsilon}(z) + z V_{\epsilon}^* T_{\epsilon}(z)^* \right) G_{\epsilon} (z)^* ,
\end{split}
\end{equation}
hence for any $\varphi \in {\mathcal H},$
\begin{equation}\label{CS}
\begin{split}
\Re \bra \varphi, H_{\epsilon} (z) \varphi \ket &\leq 2 \|\varphi \| \| G_{\epsilon} (z) \varphi \| + \|\varphi \|^2 .
\end{split}
\end{equation}

\noindent {\bf Case (Glo):} \eqref{mmt} and \eqref{ReHe} yield
\begin{equation}
\Re \bra \varphi, H_{\epsilon} (z) \varphi \ket \geq d (\epsilon, z) \| G_{\epsilon}(z) \varphi \|^2 \label{mou}
\end{equation}
for any $(\epsilon, z) \in \Omega_{\mathbb T}$. We deduce from \eqref{CS} and \eqref{mou} that
\begin{equation*}
2d (\epsilon,z ) \| G_{\epsilon}(z) \| + 1 \geq d (\epsilon, z)^2 \| G_{\epsilon}(z) \|^2
\end{equation*}
and \eqref{C0} follows as the region where $ 2X + 1-X^2$ is positive is bounded.

Inequality \eqref{mou} also implies for any $(\epsilon, z) \in \Omega_{\mathbb T},$
\begin{align*}
\| G_{\epsilon} (z) \varphi \|^2 &\leq \frac{\Re \bra \varphi, H_{\epsilon} (z) \varphi \ket}{a_0 \epsilon |z|^2} \leq \frac{2 \Re \bra \varphi, G_{\epsilon} (z) \varphi \ket}{a_0 \epsilon |z|^2}
\end{align*}
hence \eqref{G-ReG} for $\| G_{\epsilon} (z) \varphi \|$. The proof of \eqref{G-ReG} for $\| G_{\epsilon}^* (z) \varphi \|$ is analogous.
\\

\noindent {\bf Case (Loc):} \eqref{mmt0} and \eqref{ReHe} yield
\begin{equation}
\Re \bra \varphi, H_{\epsilon} (z) \varphi \ket + \frac{3\pi^2 a_1}{2 d_0^2} \epsilon |z| \|\varphi \|^2 \geq d (\epsilon, z) \| G_{\epsilon}(z) \varphi \|^2 \label{mou0}
\end{equation}
for any $(\epsilon, z) \in \Omega_0$. Taking into account \eqref{d}, \eqref{gamma:eps} and \eqref{gamma:eps1}, we note that:
\begin{equation}\label{estimate4gamma}
0\leq d (\epsilon, z) \leq 1 + a_0 \epsilon_1 \quad \mbox{and} \quad \frac{3\pi^2 a_1}{2 d_0^2} \epsilon \leq \gamma_{\epsilon} \leq \alpha \leq 1 .
\end{equation}
We deduce from \eqref{CS} and \eqref{mou0} that for any $(\epsilon, z) \in \Omega_0$,
\begin{equation*}
2d (\epsilon, z) \| G_{\epsilon}(z) \| + 2(1 + a_0 \epsilon_1) \geq d (\epsilon, z)^2 \| G_{\epsilon}(z) \|^2 ,
\end{equation*}
and \eqref{C0} follows from the region of  positivity of the polynomial $2X + 2(1 + a_0 \epsilon_1)-X^2 $.

Inequality \eqref{mou0} also yields for any $(\epsilon, z) \in \Omega_0,$
\begin{align*}
\| G_{\epsilon} (z) \varphi \|^2 &\leq \frac{\Re \bra \varphi, H_{\epsilon} (z) \varphi \ket}{a_0 \epsilon |z|^2} + \frac{3\pi^2 a_1}{2 d_0^2 a_0 |z|} \| \varphi \|^2 , \\
&\leq \frac{2 \Re \bra \varphi, G_{\epsilon} (z) \varphi \ket}{a_0 \epsilon |z|^2} + \frac{1}{a_0 \epsilon |z|^2}\left( \frac{3\pi^2 a_1 |z|}{2 d_0^2}\epsilon - 1 \right) \| \varphi \|^2 \leq \frac{2 \Re \bra \varphi, G_{\epsilon} (z) \varphi \ket}{a_0 \epsilon |z|^2}
\end{align*}
where we have finally used \eqref{ReHe} and \eqref{estimate4gamma}. This proves \eqref{G-ReG} for $\| G_{\epsilon} (z) \varphi \|$. The proof for $\| G_{\epsilon}^* (z) \varphi \|$ is analogous.

\end{proof}

Suppose either (Glo) or (Loc). For any $(\epsilon, z) \in \Omega$ and $1/2 < s \leq 1$, we define the weighted deformed resolvent:
\begin{equation}
F_{s, \epsilon} (z) := W_s (\epsilon ) G_{\epsilon} (z) W_s (\epsilon ) \label{Fs-eps}
\end{equation}
\begin{equation}
W_s (\epsilon ) := \bra A\ket^{-s} \bra \epsilon A\ket^{s -1} \label{Ws-eps}. 
\end{equation}
Note that for any $s\in (1/2,1]$, $(W_s (\epsilon ))_{\epsilon \in [0,\epsilon_0)}$ is a family of bounded selfadjoint operators. In particular, 
$\sup_{\epsilon \in [0,1]} \|W_s (\epsilon ) \| \leq 1$. Proposition \ref{firstbounds} entails:

\begin{cor}\label{secondbounds} Suppose either (Glo) or (Loc). We have for $d$ defined in \eqref{d} and \eqref{d0} resp.,
\begin{equation*}
\sup_{(\epsilon,z) \in \Omega} d (\epsilon, z) \| F_{s, \epsilon} (z) \| \leq C_0 < \infty.
\end{equation*}
In addition, given $r\in (0,1)$, for any $(\epsilon, z) \in (0,\epsilon_1]\times [r,1]\cdot {\bf S},$
\begin{align}
\max \left\{ \| G_{\epsilon} (z) W_s (\epsilon ) \varphi \|, \| G_{\epsilon} (z)^* W_s (\epsilon ) \varphi \| \right\} &\leq \sqrt{\frac{2 | \Re \bra \varphi, F_{s, \epsilon} (z) \varphi \ket |}{a_0 \epsilon r^2}} . \label{GW-ReF}
\end{align}
\end{cor}

%%%%%%%%%%%%%%%%%%%%
\subsection{Differential Inequalities}\label{diff:ineq}

Next, we derive some differential inequalities for the weighted deformed resolvents $F_{s, \epsilon}(z)$. {The first ingredient is}

\begin{prop}\label{partial=ad} Suppose (Glo) or (Loc). Then, for any $\epsilon \in (0,\epsilon_1)$, $z\in (0,1)\cdot {\bf S}$, the map $\epsilon\mapsto G_{\epsilon} (z)$ is continuously differentiable on $(0, \epsilon_1)$ in the operator norm topology with
$$
\partial_{\epsilon} G_{\epsilon} (z) = \mathrm{ad}_A G_{\epsilon} (z) + z G_{\epsilon} (z) {\mathcal Q} (\epsilon)^* G_{\epsilon} (z) ,
$$
where ${\mathcal Q}$ is defined in \eqref{Qeps}.
\end{prop}
\begin{proof} For any fixed $z\in (0,1)\cdot {\bf S}$, the map $\epsilon\mapsto G_{\epsilon} (z)$ is continuous on the interval $[0,\epsilon_1)$ {and} continuously differentiable on $(0, \epsilon_1)$ in  the operator norm topology, and
\begin{equation*}
\partial_{\epsilon} G_{\epsilon} (z) = z G_{\epsilon}(z) ( \partial_{\epsilon} S_{\epsilon}^* - \epsilon \partial_{\epsilon} B_{\epsilon}^* - B_{\epsilon}^* ) G_{\epsilon}(z) .
\end{equation*}
Now, for any $\epsilon \in [0,\epsilon_1)$, $z\in (0,1]\cdot {\bf S}$, we have $G_{\epsilon} (z) \in C^1(A)$ with
\begin{equation}\label{Q_e}
\mathrm{ad}_{A} G_{\epsilon} (z) = z G_{\epsilon} (z) ( \mathrm{ad}_A S_{\epsilon}^* - \epsilon \mathrm{ad}_A B_{\epsilon}^* ) G_{\epsilon} (z) ,
\end{equation}
which concludes the proof.
\end{proof}

Now, for any fixed $1/2 < s < 1$, the map $\epsilon \mapsto W_s (\epsilon)$ is strongly continuous on $[0,\epsilon_0)$ and converges strongly to $\bra A\ket^{-s}$ as $\epsilon$ tends to zero.

We have:

\begin{prop}\label{diffins<1} Suppose (Glo) or (Loc). Let $1/2 < s \leq 1$. For any fixed $z\in (0,1)\cdot {\bf S}$, the map $\epsilon\mapsto F_{s,\epsilon} (z)$ is weakly continuously differentiable on $(0, \epsilon_1)$ and for any $\varphi \in {\mathcal H}$, any $\epsilon \in (0,\epsilon_1)$,
\begin{align}\label{diffin.s<1}
\left| \bra \varphi, \partial_{\epsilon} F_{s, \epsilon} (z) \varphi \ket \right | & \leq h_1(\epsilon)  \|\varphi \| \sqrt{ | \bra \varphi, F_{s, \epsilon} (z)\varphi \ket | } + h_2(\epsilon) | \bra \varphi, F_{s, \epsilon} (z)\varphi \ket |.
\end{align}
where:
\begin{align}
h_1(\epsilon) &= \frac{2\sqrt{2} (2-s)\epsilon^{s-1}}{\sqrt{a_0\, \epsilon}\, r} \quad , \quad h_2(\epsilon) = \frac{2 \| {\mathcal Q} (\epsilon) \|}{a_0 r^2 \epsilon} . \label{h1:h2}
\end{align}
\end{prop}

\begin{proof} First, note that the map $\epsilon \mapsto W_s (\epsilon )$ is strongly continuously differentiable on the interval $(0,\epsilon_0)$ and that for any $\epsilon \in (0,\epsilon_0)$, any $\varphi \in {\mathcal H}$, we have
\begin{equation}\label{we}
\| \partial_{\epsilon} W_s (\epsilon ) \varphi \| \leq (1-s) \epsilon^{s -1} \| \varphi \|.
\end{equation}
Note also that for $1/2 < s \leq 1,$
\begin{equation}\label{awe}
\| AW_s (\epsilon ) \| = \| W_s (\epsilon )A \| \leq \epsilon^{s-1} .
\end{equation} 
Fix $z \in (0,1)\cdot {\bf S}$. Due to Proposition \ref{partial=ad}, the map $\epsilon\mapsto F_{s,\epsilon,} (z)$ is weakly continuously differentiable on $(0, \epsilon_0)$ and we have for any $\epsilon \in (0,\epsilon_0)$,
\begin{align*}
|\bra \varphi, \partial_{\epsilon} F_{s, \epsilon} (z) \varphi \ket | &\leq t_{1,\epsilon}(z) + | \bra W_s (\epsilon) \varphi, (\partial_{\epsilon} G_{\epsilon} (z)) W_s (\epsilon) \varphi \ket |\\
&\leq t_{1,\epsilon}(z) + t_{2,\epsilon}(z) + t_{3,\epsilon}(z)
\end{align*}
where
\begin{align*}
t_{1,\epsilon}(z) &= |\bra (\partial_{\epsilon} W_s (\epsilon ))\varphi, G_{\epsilon} (z) W_s (\epsilon) \varphi \ket + \bra W_s (\epsilon) \varphi, G_{\epsilon} (z) (\partial_{\epsilon} W_s (\epsilon )) \varphi \ket | \\
t_{2,\epsilon}(z) &= \big| \bra W_s (\epsilon) \varphi, \mathrm{ad}_A ( G_{\epsilon} (z) ) W_s (\epsilon) \varphi \ket \big| \\
t_{3,\epsilon}(z) &= \epsilon |z| |\bra {\mathcal Q}( \epsilon ) G_{\epsilon} (z)^* W_s (\epsilon) \varphi, G_{\epsilon} (z) W_s (\epsilon) \varphi \ket | .
\end{align*}
Invoking \eqref{G-ReG}, \eqref{we} and \eqref{awe}, we deduce that:
\begin{align*}
t_{1,\epsilon}(z) &\leq \| (\partial_{\epsilon} W_s (\epsilon ))\varphi \| \left( \| G_{\epsilon} (z) W_s (\epsilon) \varphi \| + \| G_{\epsilon} (z)^* W_s (\epsilon) \varphi \| \right) \\
&\leq \frac{2\sqrt{2} (1-s) \epsilon^{s-1}}{\sqrt{a_0\, \epsilon}\, r} \| \varphi \| \sqrt{| \bra \varphi, F_{s, \epsilon} (z) \varphi \ket |} , \\
t_{2,\epsilon}(z) &\leq \| A W_s(\epsilon ) \varphi \| \left( \| G_{\epsilon} (z) W_s (\epsilon) \varphi \| + \| G_{\epsilon} (z)^* W_s (\epsilon) \varphi \| \right) \\
&\leq \frac{2\sqrt{2} \epsilon^{s-1}}{\sqrt{a_0\, \epsilon}\, r} \| \varphi \| \sqrt{| \bra \varphi, F_{s, \epsilon} (z) \varphi \ket |} ,
\end{align*}
while
\begin{align*}
t_{3,\epsilon}(z) &\leq \| {\mathcal Q} (\epsilon) \| \| G_{\epsilon} (z)^* W_s (\epsilon) \varphi \| \| G_{\epsilon} (z) W_s (\epsilon) \varphi \| \leq \frac{2 \| {\mathcal Q} (\epsilon) \|}{a_0\, r^2\, \epsilon} | \bra \varphi, F_{s, \epsilon} (z) \varphi \ket | .
\end{align*}
Estimate (\ref{diffin.s<1}) follows. 
\end{proof}

The next result combines Proposition \ref{diffins<1} with Gronwall's Lemma.

\begin{prop}\label{unifbounds} Suppose (Glo) or (Loc). Fix $s\in (1/2,1]$. Then,
\begin{equation}\label{ub}
C_1 := \sup_{(\epsilon, z) \in (0,\epsilon_1] \times [r,1]\cdot {\bf S}} \| F_{s,\epsilon} (z)\| < \infty .
\end{equation}
In addition, there exists ${\mathfrak H}_{s} \in L^1((0,\epsilon_1])$, such that for any $(\epsilon, z) \in (0,\epsilon_1] \times [r,1]\cdot {\bf S},$
\begin{equation}\label{l1}
\| \partial_{\epsilon} F_{s,\epsilon} (z) \| \leq {\mathfrak H}_{s} ( \epsilon ) .
\end{equation}
\end{prop}
\begin{proof} With $h_1$ and $h_2$ defined in \eqref{h1:h2}, $h_2 \in L^1 ((0,\epsilon_1])$. By inspection, $h_1$ also belongs to $L^1((0,\epsilon_1])$. For any $(\epsilon, z) \in (0,\epsilon_1] \times [r,1]\cdot {\bf S}$ and any $\varphi \in {\mathcal H},$
$$
| \bra \varphi, F_{s, \epsilon} (z) \varphi \ket | \leq | \bra \varphi, F_{s, \epsilon_1} (z) \varphi \ket | + \int_{\epsilon}^{\epsilon_1} | \bra \varphi, \partial_{\mu} F_{s, \mu} (z) \varphi \ket |\, d\mu ,
$$
which combined with Corollary \ref{secondbounds} and Proposition \ref{diffins<1} yields:
\begin{align*}\label{int.s<1}
| \bra \varphi, F_{s, \epsilon} (z) \varphi \ket | &\leq  \frac{C_0}{a_0 \epsilon_1} \|\varphi \|^2 + \|\varphi \| \int_{\epsilon}^{\epsilon_1} h_1(\mu) \sqrt{| \bra \varphi, F_{s, \mu} (z)\varphi \ket |}\, d\mu \\
&+ \int_{\epsilon}^{\epsilon_1} h_2 (\mu) \bra \varphi, F_{s, \mu} (z)\varphi \ket |\, d\mu .
\end{align*}
Using Gronwall's Lemma as stated in e.g. \cite{abmg} Lemma 7.A.1, we deduce
\begin{align*}
| \bra \varphi, F_{s, \epsilon} (z) \varphi \ket | \leq  & \left[ \sqrt{\frac{C_0}{a_0 \epsilon_1}} + \frac{1}{2} \int_{\epsilon}^{\epsilon_1} h_1 (\mu) \exp \left( - \frac{1}{2} \int_{\mu}^{\epsilon_1} h_2(x) \, dx \right) \, d\mu \right]^2 \|\varphi \|^2 \\
& \times \exp \left( \int_{\epsilon}^{\epsilon_1} h_2 (\mu)\, d\mu \right)
\end{align*}
Since the functions $h_1$ and $h_2$ are integrable on $(0,\epsilon_1]$, we deduce that:
$$
\sup_{\| \varphi \| =1} \, \sup_{(\epsilon, z) \in (0,\epsilon_1] \times [r,1]\cdot {\bf S}} |\bra \varphi, F_{s,\epsilon} (z)\varphi \ket | < \infty .
$$
Recall that $\| F_{s,\epsilon} (z) \| = \sup_{\| \varphi \| =1, \| \psi \| =1} |\bra \varphi, F_{s,\epsilon} (z)\psi \ket |$ ; so, \eqref{ub} follows by polarisation. Incorporating it into \eqref{diffin.s<1} entails,
$$
\left| \bra \varphi, \partial_{\epsilon} F_{s, \epsilon} (z) \varphi \ket \right | \leq \left( \sqrt{C_1} h_1(\epsilon) + C_1 h_2(\epsilon) \right) \|\varphi \|^2 ,
$$
for any $(\epsilon, z) \in (0,\epsilon_1] \times [r,1]\cdot {\bf S}$ and any $\varphi \in {\mathcal H}$. \eqref{l1} follows by polarisation.
\end{proof}

\begin{prop}\label{mourre0} Suppose (Glo) or (Loc). Then for any fixed $s\in (1/2,1]$ and any fixed $z\in [r,1)\cdot {\bf S},$
\begin{equation}\label{weaklimit}
w- \lim_{\epsilon \rightarrow 0^+} F_{s,\epsilon} (z) = F_s (z)
\end{equation}
where
\begin{equation}\label{Fs(z)}
F_s (z) = \bra A\ket^{-s} (1-zV^*)^{-1} \bra A\ket^{-s} \quad \text{for }\, z\in {\mathbb D} .
\end{equation}
\end{prop}
\begin{proof} 

Fix $r\in (0,1)$. Pick two vectors $\varphi$, $\psi$ in ${\mathcal H}$ and fix for a moment $z \in [r,1)\cdot {\bf S}$. For any $\epsilon \in (0,\epsilon_0)$,
$$
\bra \varphi, \bra A\ket^{-s} (1-zV^*)^{-1} \bra A\ket^{-s} \psi \ket - \bra \varphi, F_{s,\epsilon}(z) \psi \ket = t_{1, \epsilon} (z) + t_{2, \epsilon} (z) + t_{3, \epsilon} (z)
$$
with
\begin{align*}
t_{1, \epsilon} (z) &= \bra W_s (\epsilon) \varphi, ( (1-zV^*)^{-1} - G_{\epsilon} (z) ) W_s (\epsilon) \psi \ket \\
t_{2, \epsilon} (z) &= \bra W_s (0) \varphi, (1-zV^*)^{-1} ( W_s (0) -W_s (\epsilon) ) \psi \ket \\
t_{3, \epsilon} (z) &= \bra ( W_s (0) - W_s (\epsilon) ) \varphi, (1-zV^*)^{-1} W_s (\epsilon) \psi \ket .
\end{align*}
First, we estimate $t_{1, \epsilon} (z)$. For any $(\epsilon, z) \in [0,\epsilon_1] \times (0,1)\cdot {\bf S}$, $G_{\epsilon} (z) - (1-zV^*)^{-1} = G_{\epsilon}(z) ( T_{0} (z) - T_{\epsilon} (z) ) (1-zV^*)^{-1}$. We observe that for $z\in {\mathbb D}$, we have $\| (1-zV^*)^{-1} \| \leq 1/(1-|z|)$. Using Lemma \ref{Tepsilon-T0} and Proposition \ref{firstbounds}, we deduce that
\begin{equation}
\| G_{\epsilon} (z) - (1-zV^*)^{-1} \| \leq \min \left\{ \frac{C_0 b \, \epsilon}{(1+|z|)( 1-|z| )^2},\frac{C_0 b}{a_0 |z|^2 ( 1-|z| )} \right\}  \label{Geps-G0}
\end{equation}
hence
$$
| t_{1, \epsilon} (z) | \leq \frac{C_0 b \, \epsilon}{(1+|z|)( 1-|z| )^2} \|\varphi \| \|\psi \|
$$
since $\| W_s (\epsilon) \| \leq 1$ for any $\epsilon \in [0,1]$. Now, since $\| (1-zV^*)^{-1} \| \leq 1/(1-|z|)$ for any $z\in {\mathbb D}$, we also have
\begin{align*}
| t_{2, \epsilon} (z) | &\leq \frac{1}{1-|z|} \| (W_s (0)-W_s ( \epsilon))\varphi \| \|\psi \| \\
| t_{3, \epsilon} (z) | &\leq \frac{1}{1-|z|} \|\varphi \| \| (W_s (0)-W_s (\epsilon))\psi \| ,
\end{align*}
which concludes.

\end{proof}

We have arrived at the {heart} of the proof of Theorems \ref{lap1} and \ref{lap3}:
\begin{proposition}\label{mourre1} Suppose (Glo) or (Loc). For $s\in (1/2,1]$. Then, $F_s$ defined by \eqref{Fs(z)} admits a continuous extension to ${\mathbb D} \cup {\bf S}$. In addition, given $r\in (0,1)$, it holds
\begin{equation}\label{Fs-esp-Fs}
\sup_{z\in [r,1)\cdot {\bf S}} \| F_{s, \epsilon}(z) - F_s (z) \| \rightarrow _{\epsilon \rightarrow 0^+} 0 .
\end{equation}
We denote the extension also by $F_s$.
\end{proposition}
\begin{proof} From Proposition \ref{unifbounds}, we get for any $\epsilon, \mu \in (0,\epsilon_1]$, $\epsilon < \mu$, $z \in [r,1)\cdot {\bf S},$
$$
\| F_{s, \epsilon}(z) - F_{s, \mu}(z) \| \leq \int_{\epsilon}^{\mu} {\mathfrak H}_{s} (\xi ) d\xi
$$
which vanishes as $\mu$ and $\epsilon$ tend to $0$ (uniformly in $z\in [r,1]\cdot {\bf S}$). So, $(F_{s,\epsilon,}(z))$ is Cauchy in ${\mathcal B}({\mathcal H})$, hence converges to some $F_s^+ (z) \in {\mathcal B}({\mathcal H})$ for any fixed $z\in [r,1]\cdot {\bf S}$. In view of Proposition \ref{mourre0}, $F_s^+ (z) = F_s (z)$ if $z\in [r,1]\cdot {\bf S}$ and
\begin{equation}\label{ing0}
\| F_{s, \epsilon}(z) - F_s^+ (z) \| \leq \int_0^{\epsilon} {\mathfrak H}_{s} (\xi ) d\xi
\end{equation}
for any $\epsilon \in (0,\epsilon_1]$. \eqref{Fs-esp-Fs} follows. It remains to prove the continuity of $F_s^+$. Given $\epsilon \in (0,\epsilon_1]$, $z,z' \in [r,1]\cdot {\bf S},$
\begin{eqnarray}
F_{s, \epsilon}(z) - F_{s, \epsilon}(z') &=& W_s (\epsilon) (G_{\epsilon}(z) - G_{\epsilon}(z')) W_s (\epsilon) \nonumber\\
&=& (z-z')  W_s (\epsilon) G_{\epsilon}(z) (S_{\epsilon}^* -\epsilon B_{\epsilon}^*) G_{\epsilon}(z') W_s (\epsilon) \label{ing1}
\end{eqnarray}
where we have used \eqref{Teps(z)} in the final step.

Now, following Corollary \ref{secondbounds} and Proposition \ref{unifbounds}, we get for any $\epsilon \in (0,\epsilon_1]$ and any $z\in [r,1]\cdot {\bf S},$
\begin{align*}
\max \left\{ \| G_{\epsilon}(z)^* W_s(\epsilon) \|, \| G_{\epsilon}(z) W_s(\epsilon) \| \right\} &\leq \frac{B_r}{\sqrt{\epsilon}} \quad \text{with} \quad B_r := \frac{\sqrt{2C_1}+1}{r\sqrt{a_0}} .
\end{align*}
Combining these estimates with \eqref{ing1} entails for any $\epsilon \in (0,\epsilon_1]$, $z,z' \in [r,1]\cdot {\bf S}$,
\begin{align*}
\| F_{s, \epsilon}(z) - F_{s, \epsilon}(z') \| \leq \frac{B_r^2}{\epsilon} |z-z'| . 
\end{align*}
Due to \eqref{ing0}, we obtain
\begin{equation*}
\begin{split}
\| F_s^+ (z) - F_s^+ (z') \| &\leq \| F_s^+ (z) - F_{s, \epsilon}(z) \| +\| F_{s, \epsilon}(z) - F_{s, \epsilon}(z') \| +\| F_s^+ (z') - F_{s, \epsilon}(z') \| \\
&\leq 2\int_0^{\epsilon} {\mathfrak H}_{s} (\xi ) d\xi + \frac{B_r^2}{\epsilon} |z-z'| , 
\end{split}
\end{equation*}
for any $\epsilon \in (0,\epsilon_1]$. The conclusion follows since the function {${\mathfrak H}_{s}$} is integrable.
\end{proof}

%%%%%%%%%%%%%%%%%%%%
\subsection{Proof of Theorem \ref{lap3}}

Rephrasing Proposition \ref{mourre1} under assumptions (Glo), we have shown that the weighted resolvent $\bra A\ket^{-s} (1-zV^*)^{-1} \bra A\ket^{-s}$, which is defined naturally for $z\in {\mathbb D}$, extends continuously to $\overline{\mathbb D}$.

%%%%%%%%%%%%%%%%%%%%
\subsection{Proof of Theorem \ref{lap1}}

We rephrase Proposition \ref{mourre1} under assumptions (Loc). Given $\theta$ such that $e^{i\theta} \in e^{i\Theta}\setminus {\mathcal E}(U)$, we have shown there exists an open neighborhood of $\theta$, $\Theta_0 \subset \Theta$ such that the weighted resolvent $\bra A\ket^{-s} (1-zV^*)^{-1} \bra A\ket^{-s}$, which is defined a priori for $z\in {\mathbb D}$, extends continuously to $\overline{\mathbb D}\cup e^{i\Theta_0}$. Since the choice of $\theta$ was arbitrary, we have actually shown that this continuous extension holds on ${\mathbb D} \cup e^{i\Theta}\setminus {\mathcal E}(U)$.

%%%%%%%%%%%%%%%%%%%%%%%%%%%%%%%%%%%%%%%%%%%%
%%%%%%%%%%%%%%%%%%%%%%%%%%%%%%%%%%%%%%%%%%%%

\section{Proof of Theorem \ref{lap2}}\label{prooflap2}

 In Theorem \ref{lap1}, the estimates only hold for $z\in {\mathbb D}$ with $r \leq |z|<1$ and $\arg z \in \Theta$. In order to obtain an estimate for all $z\in {\mathbb D}$ and to derive Theorem \ref{lap2}, we need to localize spectrally. We will do this in Lemmata \ref{AchiA} and \ref{(1-Q)G0}, then proceed to the proof of Theorem \ref{lap2}.
\begin{lemma}\label{AchiA} Let $U$ be a unitary operator which belongs to $C^1(A)$ and $\phi \in C^1({\mathbb T})$ be a function such that its derivative $\phi'$ belongs to the Wiener algebra. With $\Phi$ defined by $\Phi (e^{i\theta}) = \phi (\theta)$ for $\theta \in {\mathbb T}$, the operator $\bra A\ket^{-\delta} \Phi(U) \bra A\ket^{\delta}$ extends to a bounded operator in ${\mathcal B}({\mathcal H})$ for all $\delta \in [-1,1]$.
\end{lemma}

\begin{proof} By hypothesis, we have that: $\Phi (e^{i\theta}) = \phi (\theta) = \sum_{n\in {\mathbb Z}} \hat{\phi}_n e^{in\theta}$ with $\sum_{n\in {\mathbb Z}} |n \hat{\phi}_n| < \infty$. We first prove that $\Phi(U) \in C^1(A)$. First, observe that for any $n\in {\mathbb Z}$, $U^n \in C^1(A)$ with: 
$$
\mathrm{ad}_A U^n = \sum_{k=0}^{n-1} U^k (\mathrm{ad}_A U) U^{n-1}
$$
for $n\geq 1$ and $\mathrm{ad}_A U^n = - U^n (\mathrm{ad}_A U^{-n}) U^n$ for $n\leq -1$. So, $\| \mathrm{ad}_A U^n \| \leq |n| \| \mathrm{ad}_A U \|$ for any $n\in {\mathbb Z}$. Given $N\in {\mathbb Z}_+$, let
$$
\Phi_N (U) = \sum_{|n| \leq N} \hat{\phi}_n U^n .
$$
$(\Phi_N (U))_{N\in {\mathbb N}} \subset C^1(A)$ converges in norm to $\Phi (U)$ and $(\mathrm{ad}_A \Phi_N (U))_{N\in {\mathbb N}}$ is convergent w.r.t the operator norm topology. So, $\Phi (U)\in C^1(A)$ and $\lim_{N\rightarrow \infty} \mathrm{ad}_A \Phi_N (U) = \mathrm{ad}_A \Phi (U)$ \cite{abmg, ggm}.

As a consequence, we deduce that for $\delta \in \{-1,0,1\}$, $\bra A\ket^{-\delta} \Phi(U) \bra A\ket^{\delta}$ extend to a bounded operator in ${\mathcal B}({\mathcal H})$. We conclude by interpolation.
\end{proof}

\begin{lemma}\label{(1-Q)G0} For any $z\in {\mathbb D}$ and any $\psi \in {\cal H}$, it holds:
\begin{align*}
\| (U^* -V^*) G_0(z) \psi \|^2 &\leq 8 \bra \psi , \Re (G_0 (z)) \psi \ket .
\end{align*}
\end{lemma}
\begin{proof} We decompose $U^* -V^* = U^* \overline{P} + \overline{Q}U^* P$, so
\begin{align*}
\| (U^* -V^*) G_0 (z) \psi \|^2 &\leq 2 \| \overline{P} G_0 (z) \psi \|^2 + 2 \| \overline{Q}U^* P G_0 (z) \psi \|^2 .
\end{align*}
It remains to bound each term on the RHS. Let $\overline{P} = 1-P$. Since $0 \leq P \leq 1$ and $0 \leq Q \leq 1$, so $P^2 \leq P$ and
\begin{align*}
\overline{P}^2 &\leq 1- P^2 = 1 - PUU^*P \leq 1 - PUQ^2 U^*P = 1-VV^* \\
&\leq 1- |z|^2 VV^* = T_0 (z) + T_0 (z)^* zV^* ,
\end{align*}
for any $z\in {\mathbb D}$. It follows that for any $\psi \in {\cal H},$
\begin{align*}
\| \overline{P} G_0 (z)\psi \|^2 &\leq \bra G_0 (z)\psi, \psi \ket + \bra \psi, zV^* G_0 (z) \psi \ket \\
&= \bra \psi, 2\Re (G_0 (z))\psi \ket - \| \psi \|^2 \leq 2 \bra \psi, \Re (G_0 (z))\psi \ket .
\end{align*}
Also $\overline{Q}^2 \leq 1-Q^2$, hence it holds
\begin{align*}
PU\overline{Q}^2 U^* P &\leq PU(1-Q^2) U^* P  \leq 1 - PUQ^2 U^*P \leq 1 - VV^* \\
&\leq 1- |z|^2 VV^* = T_0 (z) + T_0 (z)^* zV^* ,
\end{align*}
for any $z\in {\mathbb D}$. As before, we get for any $\psi \in {\cal H},$
\begin{align*}
\| \overline{Q}U^* P G_0(z) \psi \|^2 &\leq \bra G_0 (z)\psi, \psi \ket + \bra \psi, zV^* G_0 (z) \psi \ket \leq 2 \bra \psi, \Re (G_0 (z))\psi \ket .
\end{align*}
and the conclusion follows.
\end{proof}

\begin{proof} (of Theorem \ref{lap2})

Let $\Theta_1$ and $\Theta_0$ be open connected subsets such that $\overline{\Theta_1} \subset \Theta_0 \subset \overline{\Theta_0} \subset \Theta$. Let $\phi \in C^{\infty}({\mathbb T}; {\mathbb R})$ be supported in $\overline{\Theta_1}$. Combining Theorem \ref{lap1} with Lemma \ref{AchiA} yields:
\begin{equation*}
\sup_{|z|<1, \arg z \in \Theta_0} \| \bra A\ket^{-s} \Phi (U)(1-zV^*)^{-1} \Phi (U) \bra A\ket^{-s} \| < \infty .
\end{equation*}
For $z\in {\mathbb D}$, the resolvent identity reads:
\begin{align*}
(1-zV^*)^{-1} &= (1-zU^*)^{-1} - z (1-zU^*)^{-1} (U^* - {V^*}) (1-zV^*)^{-1} .
\end{align*}
Given any $z\in {\mathbb D}$, we bound the quantity $\| \bra A\ket^{-s} \Phi(U) (1-zV^*)^{-1} \Phi(U) \bra A\ket^{-s} \|$ by:
\begin{align*}
\| \bra A\ket^{-s} &\Phi(U) (1-zU^*)^{-1} \Phi(U) \bra A\ket^{-s} \| \\
&+ \| (U^* -{V^*}) (1-zV^*)^{-1} \Phi(U) \bra A\ket^{-s} \| \|(1-\bar{z}U)^{-1}\Phi(U) \bra A\ket^{-s} \| .
\end{align*}
Unitary functional calculus yields:
\begin{align*}
C_{\phi, \Theta_0} &= \sup_{|z| <1,\, \arg z\, \in\, {\mathbb T} \setminus \Theta_0} \| (1-\bar{z}U)^{-1} \Phi (U)\| < \infty .
\end{align*}
From Lemma \ref{(1-Q)G0}, we deduce that for any $z\in {\mathbb D}$ with $\arg z\in {\mathbb T}\setminus \Theta_0$, the quantity $\| \bra A\ket^{-s} \Phi(U) (1-zV^*)^{-1} \Phi(U) \bra A\ket^{-s} \|$ is bounded by:
$$
C_{\phi, \Theta_0} \left( \|\Phi \|_{\infty} + 2\sqrt{2} \sqrt{\| \bra A\ket^{-s} \Phi(U) (1-zV^*)^{-1} \Phi(U) \bra A\ket^{-s} \| } \right) ,
$$
which entails
\begin{equation*}
\sup_{|z|<1, \arg z \in {\mathbb T}\setminus \Theta_0} \| \bra A\ket^{-s} \Phi (U)(1-zV^*)^{-1} \Phi (U) \bra A\ket^{-s} \| < \infty .
\end{equation*}
This completes the proof.

\end{proof}

%%%%%%%%%%%%%%%%%%%%%%%%%%%%%%%%%%%%%%%%%%%%
%%%%%%%%%%%%%%%%%%%%%%%%%%%%%%%%%%%%%%%%%%%%
\section{Appendix}\label{app}

\begin{proposition}\label{mourre-equiv} Let $B\in\mathcal{B}({\mathcal H})$ be symmetric and $E$ an orthogonal projection acting on ${\mathcal H}$. Denote $E^\bot :=1-E$. The following statements are equivalent: there exist $c>0$, $a\in (0,c]$ and $b>0$ such that:
\begin{enumerate}
\item[(a)] $EBE\geq c E$,
\item[(b)] $B\geq a E- b E^\bot= a - (a+b) E^\bot$.
\end{enumerate}
\end{proposition} 
\begin{proof} Assume (a). Write $B= EBE + 2\Re (EBE^{\perp}) +E^{\perp}BE^{\perp}$ and note that: $E^\bot BE^\bot \geq -\|B\| E^\bot$ and $EBE \geq c E$ by hypothesis. In addition, for any $p>0$
$$
2\Re (EB E^\bot ) \leq p \|B \| E + p^{-1} \|B \| E^\bot ,
$$
so
$$
B \geq (c- p\|B \|) E - \|B \| (1+p^{-1}) E^\bot .
$$
Once fixed $p>0$ in such a way that $0 < c- p\|B \| < c$, we obtain (b). The converse implication is immediate.
\end{proof}

\begin{remarks}
\begin{enumerate}
\item In applications it is convenient {to note} that: for unitary $U$ one has $U\in C^1(A)$ if an only if
\begin{itemize}
\item[] there exists a core ${\mathcal S}$ for $A$ such that $U{\mathcal S} \subset {\mathcal S}$ and  $U^*AU-A:{\mathcal S}\rightarrow {\mathcal H}$ extends to a bounded operator on ${\mathcal H}$ ; or, equivalently:
 $U^*{\mathcal S} \subset {\mathcal S}$ and  $A-UAU^*:{\mathcal S}\rightarrow {\mathcal H}$ extends to a bounded  operator.
\end{itemize}
In addition, the bounded extension of $U^*AU-A$  is precisely $U^* (\mathrm{ad}_A U)$. See Propositions \ref{6} and \ref{7} for details.
\end{enumerate}
\end{remarks}

\begin{prop}\label{6} If $U\in C^1(A)$ is unitary, then $U {\mathcal D}(A)\subset {\mathcal D}(A)$ and the operator $U^*AU-A:{\mathcal D}(A)\rightarrow {\mathcal H}$ extends to a bounded operator on ${\mathcal H}$, denoted $(U^*AU-A)$. It holds: $(U^*AU-A) = U^*\mathrm{ad}_A U = -(\mathrm{ad}_A U^*)U$. Conversely, let ${\mathcal S}$ be a core for $A$ such that $U{\mathcal S} \subset {\mathcal S}$. Assume that the operator $U^*AU-A:{\mathcal S}\rightarrow {\mathcal H}$ extends to a bounded operator on ${\mathcal H}$ and denote by $C$ this extension. Then, $U$ (and $U^*$) belongs to $C^1(A)$ and $U^* \mathrm{ad}_A U=C$.
\end{prop}
\noindent{\bf Proof:} For the proof of the first statement, see \cite{ggm} Proposition 2.2. We deduce that for all $(\varphi, \psi)\in {\mathcal H}\times {\mathcal D}(A)$, $\bra \varphi, U^*AU\psi\ket - \bra \varphi, A\psi\ket =\bra \varphi, U^*(AU-UA)\psi\ket$. This implies our claim. Conversely, assume that for all $(\varphi, \psi)\in {\mathcal S}\times {\mathcal S}$, $\bra A\varphi, U\psi\ket -\bra \varphi, UA\psi\ket = \bra \varphi, AU\psi\ket -\bra U^*\varphi, A\psi\ket = \bra U^*\varphi, U^*AU\psi\ket -\bra U^*\varphi, A\psi\ket = \bra U^*\varphi, C\psi\ket = \bra \varphi, U C\psi\ket$. The identity extends continuously over ${\mathcal D}(A) \times {\mathcal D}(A)$. This shows that $U\in C^1(A)$ and that: $\mathrm{ad}_A U=UC$. \ep

\begin{rem}\label{7} Analog relations can be built between $A-UAU^*$ and $(\mathrm{ad}_A U) U^* = - U \mathrm{ad}_A U^*$.
\end{rem}

%%%%%%%%%%%%%%%%%%%%
%%%%%%%%%%%%%%%%%%%%

\end{document}